\documentclass[11pt,reqno]{amsart}
\allowdisplaybreaks[4]
\usepackage{graphicx} 
\usepackage{subfigure}
\usepackage{tcolorbox}
\usepackage{epsfig}
\usepackage{amssymb}
\usepackage{xypic}
\usepackage{cite,color}
\newtheorem{theorem}{Theorem}

\newtheorem{proposition}{Proposition}

\newtheorem{corollary}{Corollary}

\newtheorem{lemma}{Lemma}
\newtheorem{remark}{Remark}
\newcommand{\be}{\begin{equation}}
\newcommand{\ee}{\end{equation}}
\newcommand{\bea}{\begin{eqnarray}}
\newcommand{\eea}{\end{eqnarray}}
\newcommand{\ba}{\begin{array}}
\newcommand{\ea}{\end{array}}
\newcommand{\bean}{\begin{eqnarray*}}
\newcommand{\eean}{\end{eqnarray*}}

\newcommand{\pa}{\partial}

\begin{document}

\title{Bosonic construction of CKP tau function}
\author{Shen Wang$^{1}$, Wenchuang Guan$^{1}$, Jipeng Cheng$^{1,2*}$}
\dedicatory {$^{1}$ School of Mathematics, China University of
Mining and Technology, \ Xuzhou, Jiangsu 221116, China,\\
$^{2}$ Jiangsu Center for Applied Mathematics (CUMT), \ Xuzhou, Jiangsu 221116, China.
}

\thanks{*Corresponding author. Email: chengjp@cumt.edu.cn.}
\begin{abstract}
The CKP tau function has been an important topic in mathematical physics. In this paper, the inverse of vacuum expectation value of exponential of certain bosonic fields, is showed to be the CKP tau function given by Chang and Wu, in the language of CKP Darboux transformation. In fact, computation of the above vacuum expectation value is usually quite difficult, since the square of bosonic fields is usually not zero. Here the corresponding vacuum expectation value is understood as successive application of CKP Darboux transformations, so that we can compute it by using the methods of integrable systems, where a useful formula is given. For applications, we construct solutions of KdV hierarchy by vacuum expectation value of bosonic fields, by the fact that KdV hierarchy is the 2--reduction of CKP hierarchy.
\end{abstract}
\maketitle

\tableofcontents

\section{Introduction}
\subsection{Backgrounds}
\subsubsection{KP hierarchy}
KP hierarchy \cite{Date1983WorldSci} has been playing an important role in mathematical and theoretical physics, which is defined by
\begin{align}\label{KPLax-equ}
L_{t_n}=[(L^n)_{\geq 0},L],\  n=1,2,\cdots,
\end{align}
with Lax operator
$$L=\pa+\sum_{j=1}^{\infty}u_{j+1}(t)\pa^{-j},$$
and $\pa=\pa_x$, $t=(t_1=x,t_2,t_3,\cdots)$, $[A,B]=AB-BA$ and $(\sum_i a_i \pa^i)_{\geq 0}=\sum_{i\geq 0} a_i \pa^i$.
KP hierarchy has many different formulations. Among them, the famous one is given by the tau function, that is
\begin{align}\label{KP-bilinintro-tau}
{\rm Res}_\lambda\tau_{{\rm KP}}(t-[\lambda^{-1}])\tau_{{\rm KP}}(t'+[\lambda^{-1}])e^{\widetilde{\xi}(t-t',\lambda)}=0,
\end{align}
where $[\lambda^{-1}]=(\lambda^{-1},\frac{\lambda^{-2}}{2},\cdots)$ and $\widetilde{\xi}(t,\lambda)=\sum^{\infty}_{n=1} t_n \lambda^n$.
If introduce wave function $\psi_{{\rm KP}}(t,\lambda)$ and adjoint wave function $\psi_{{\rm KP}}^*(t,\lambda)$ by the way below
\begin{align}\label{KP-psi-tau}
\psi_{{\rm KP}}(t,\lambda)=\frac{\tau_{{\rm KP}}(t-[\lambda^{-1}])}{\tau_{{\rm KP}}(t)} e^{\widetilde{\xi}(t,\lambda)},\ \
\psi^*_{{\rm KP}}(t,\lambda)=\frac{\tau_{{\rm KP}}(t+[\lambda^{-1}])}{\tau_{{\rm KP}}(t)} e^{-\widetilde{\xi}(t,\lambda)},
\end{align}
then one can find
\begin{align}\label{KP-bili-psi}
{\rm Res}_\lambda\psi_{{\rm KP}}(t,\lambda)\psi^*_{{\rm KP}}(t',\lambda)=0.
\end{align}
Next if introduce a dressing operator $$W=1+w_1\pa^{-1}+w_2\pa^{-2}+w_3\pa^{-3}+\cdots,$$
such that $\psi_{{\rm KP}}(t,\lambda)=W(e^{\widetilde{\xi}(t,\lambda)})$, then
\begin{align*}
W_{t_n}=-(W\partial^nW^{-1})_{<0}W,
\end{align*}
and $W\pa W^{-1}$ is the required KP Lax operator $L$ in \eqref{KPLax-equ}, that is $L=W\pa W^{-1}$. Further by \eqref{KP-psi-tau}, the coefficients of $L$ can be expressed by KP tau function in the way below
\begin{align*}
u_2&=\pa^2_x\log \tau_{{\rm KP}},\\ u_3&=\frac{1}{2}(\pa_x\pa_{t_2}-\pa_x^3)\log \tau_{{\rm KP}},\\ u_4&=\frac{1}{6}\Big(2\pa_x\pa_{t_3}+\pa_x^4-3\pa_x^2\pa_{t_3}-\pa_x^2\Big)\log \tau_{{\rm KP}}.
\end{align*}

KP tau function has many different forms (see \cite{Harnad2021} and its references). For example, Schur polynomials $p_\lambda(t)$ \cite{Harnad2021,satorims,willox2004} are always KP tau functions. In Kyoto school's work, tau function is viewed as one point of infinite dimensional Grassmannian \cite{Harnad2021,satorims,willox2004}. KP tau function can be written into the linear combination of Schur polynomials, that is
\begin{align*}
\tau_{{\rm KP}}(t)=\sum_{\lambda}c_{\lambda}p_{\lambda}(t).
\end{align*}
Here $c_{\lambda}$ can be viewed as the Pl$\ddot{u}$cker coordinates of $\tau_{{\rm KP}}(t)$ satisfying the Pl$\ddot{u}$cker relations \cite{Harnad2021,satorims,willox2004}
\begin{align*}
\sum_{k\geq0}(-1)^{k+l}c_{(\alpha_1-1,\cdots\alpha_l-1,\beta_{k+1}-k+l+1,\alpha_{l+1},\cdots)}\cdot c_{(\beta_1+1,\cdots,\beta_k+1,\beta_{k+2},\cdots)}=0,
\end{align*}
where $\alpha=(\alpha_1,\alpha_2,\cdots),\  \beta=(\beta_1,\beta_2,\cdots)$ are partitions,
which is another equivalent formulation for bilinear equation \eqref{KP-bilinintro-tau}.
Other famous KP tau functions (see \cite{Harnad2021} and its references),  include elliptic functions, multi--soliton solution, rational solution, quasi--periodic solution related to algebraic curves, matrix model integrals, Hurwitz numbers, Kontsevich--Witten tau function and so on.
Every different manifestation of KP tau functions stands for the corresponding application of KP theory in mathematical and theoretical physics.

Due to the great success for KP theory, one expects that other KP theories, such as KP hierarchy of B and C types (BKP and CKP for short), also have many applications. BKP hierarchy is given by imposing BKP constraints on KP Lax operator in the way below \cite{Date1981JPAJ,Zabrodin2021,Date1983WorldSci}
\begin{align*}
L^*=-\pa L\pa^{-1},
\end{align*}
where $(\sum_i a_i \pa^i)^*=\sum_{i} (-\pa)^i a_i$.
BKP hierarchy has almost the same structure as the usual KP hierarchy, since it has its own tau functions $\tau_{{\rm BKP}}$, which is related with wave function $\psi_{{\rm BKP}}$ by
\begin{align*}
\psi_{{\rm BKP}}(t_o,\lambda)=\frac{\tau_{{\rm BKP}}(t_o-2[\lambda^{-1}]_o)}{\tau_{{\rm BKP}}(t_o)} e^{\xi(t_o,\lambda)},
\end{align*}
where $t_o=(t_1,t_3,t_5,\cdots),\ [\lambda^{-1}]_o=(\lambda^{-1},\frac{\lambda^{-3}}{3},\cdots)$,  $\xi(t_o,\lambda)=\sum_{n=0}^{\infty}t_{2n+1}\lambda^{2n+1}$.
The famous example of BKP tau functions is Schur--$Q$ functions \cite{youbkp,nimmojpa}, by which $\tau_{{\rm BKP}}$ can be expanded in terms of Cartan coordinates  satisfying Cartan relations \cite{Harnad2021} instead of the Pl$\ddot{u}$cker coordinates and Pl$\ddot{u}$cker relations in KP theory.
BKP theory also has many applications in mathematical physics and theoretical physics, such as the Brezin--Gross--Witten theory \cite{Liuxiaobo2022,Alexandrov2023}. The reason why BKP theory can be so successful as KP theory is the BKP's own tau function $\tau_{\rm BKP}$.

\subsubsection{CKP hierarchy}\label{subsection:CKP}
Different from KP and BKP theories,
the study of CKP hierarchy is quite difficult and challenging, since the definition of CKP own tau function is a big barrier. CKP hierarchy is also a KP sub--hierarchy \cite{Date1981JPAJ,Zabrodin2021} defined by
\begin{align*}
L^*=-L,
\end{align*}
and the corresponding bilinear equation is given by the wave function $\psi(t_o,\lambda)$ as follows,
\begin{align}\label{CKP-bilinintro-psi}
{\rm Res}_\lambda\psi(t_o,\lambda)\psi(t'_o,-\lambda)=0.
\end{align}
As a KP sub--hierarchy, CKP has a tau function inherited from the KP hierarchy, that is
\begin{align}\label{CKP-wave-tau-KP}
\psi(t_o,\lambda)=\frac{\tau_{{\rm KP}}(\dot{ t}-[\lambda^{-1}])}{\tau_{{\rm KP}}(\dot{t})} e^{\xi(t_o,\lambda)},
\end{align}
where $\dot{t}=(t_1,0,t_3,0,\cdots)$.
However, this tau function is not quite applicable, due to the shift of even time part $t_e=(t_2,t_4,\cdots)$ in \eqref{CKP-wave-tau-KP}. In a long time, peoples hope that CKP also has its own tau function just like BKP.

In \cite{Date1981JPAJ}, Date, Jimbo, Kashiwara and Miwa (DJKM) proposed that CKP tau function can be constructed by free bosons
\begin{eqnarray}\label{tau-vac-ht}
\tau_{{\rm CKP}}(t_o)=\langle 0| e^{H(t_o)}g |0\rangle^{-2},
\end{eqnarray}
where
$H(t_o)=\frac{1}{2}\sum_{n\geq0} t_{2n+1}\sum_{j\in\mathbb{Z}+1/2}(-1)^{j+\frac{1}{2}}\phi_j\phi_{-j-2n-1}$ and
$g\in Sp_\infty$ with
\begin{align}\label{intro:sp}
Sp_{\infty}=\Big\{e^{A_1}e^{A_2}\cdots e^{A_k}\biggr|A_{l}=\sum_{i,j\in \mathbb{Z}+\frac{1}{2}}a_{l,ij}:\phi_{i}\phi_{j}:,
\ \text{$a_{l,ij}=0$ for $|i+j|\gg0$}\Big\}.
\end{align}
Here $\phi_i$ ($i\in\mathbb{Z}+1/2$) is the free boson satisfying \cite{vandeleur2012}
\begin{eqnarray*}
\phi_i\phi_j-\phi_j\phi_i=(-1)^{j-\frac{1}{2}}\delta_{i,-j}.
\end{eqnarray*}
Vacuum vectors $|0\rangle$ and $\langle 0|$ are defined by
\begin{align*}
\phi_{-j}|0\rangle=0,\quad\langle 0|\phi_{j}=0,\quad j>0,
\end{align*}
and normal ordering is given by
\begin{align*}
:\phi_i\phi_j:=\left\{
\begin{array}{ll}
\phi_i\phi_j, & \hbox{${\rm if}\ i\geq j$,} \\
\phi_j\phi_i, & \hbox{${\rm if}\ j>i$.}
  \end{array}
\right.
\end{align*}

Following DJKM's thought, CKP hierarchy is constructed by van de Leur, Orlov and Shiota in \cite{vandeleur2012} as a sub--hierarchy of super BKP hierarchy \cite{kac1989}. In this case, there is a unique tau function $\tau(t_o,t_s)$ depending on super variables $t_s=(t_{\frac{1}{2}},t_{\frac{3}{2}},\cdots)$, which can be expanded by the way below
\begin{align*}
\tau(t_o,t_s)=\tau_0(t_o)+\sum_{\nu\in\mathbb{Z}^{>1}_{odd}}\tau_{(\nu,1)}(t_o)t_{\frac{\nu}{2}}t_{\frac{1}{2}}+{\rm terms\ of\ at\ least\ two }\ t_s.
\end{align*}
Then the CKP wave function $\psi(t_o,\lambda)$ can be written as
\begin{align}\label{vandeleur-tau}
\psi(t_o,\lambda)=\frac{\tau_0(t_o+2[\lambda^{-1}]_o)+ \sum\limits_{\nu\in\mathbb{Z}^{>1}_{odd}}\tau_{(\nu,1)}(t_o+2[\lambda^{-1}]_o)\lambda^{-\frac{\nu+1}{2}}}
{\tau_0(t_o)}e^{\xi(t_o,\lambda)},
\end{align}
where $\tau_0(t_o)$ and $\tau_{(\nu,1)}(t_o)$ can be expressed by using free bosons
\begin{align*}
\tau_0(t_o)=\langle 0 | e^{H(t_o)}g|0\rangle,\ \tau_{(\nu,1)}(t_o)=4\langle \nu,1|e^{H(t_o)}g|0\rangle,\ \cdots,
\end{align*}
with $\langle \nu,1|=\langle 0|H_{\frac{\nu}{2}}H_{\frac{1}{2}}$ and $H_{\frac{\nu}{2}}$ in \eqref{hexpression}.
Therefore by \eqref{tau-vac-ht}, one can find that
\begin{align}\label{tau-CKP-0}
\tau_{\rm CKP}(t_o)=\tau_0(t_o)^{-2}.
\end{align}
Further for CKP Lax operator $L=\partial+\sum_{j=1}u_{j+1}\partial^{-j}$, we can find by \eqref{vandeleur-tau} that
\begin{align*}
u_2(t_o)=2\partial_x^2\log(\tau_0(t_o))=-\partial_x^2\log(\tau_{\rm KP}(\dot{t})),
\end{align*}
which reminds us that $\tau_0(t_o)^{-2}$ should be the tau function
$\tau_{\rm KP}(\dot{t})$ of CKP hierarchy inherited from KP tau function, that is,
\begin{align}\label{tau-KP-0}
\tau_{\rm KP}(\dot{t})=\tau_0(t_o)^{-2}.
\end{align}
Thus by (\ref{tau-CKP-0}), we can believe
$$\tau_{{\rm CKP}}(t_{o})=\tau_{\rm KP}(\dot{t}).$$
So $\tau_{{\rm CKP}}(t_{o})$ is still not the real CKP's own tau function.

Fortunately in 2013, Chang and Wu (CW) succeeded in  introducing CKP's own tau function $\tau_{{\rm CW}}(t_o)$ by the way below \cite{chang2013}
\begin{align}\label{wuchaozhong-tau}
\psi(t_o,\lambda)&=\left(1+\frac{1}{\lambda}\pa_x{\rm log}\frac{\tau_{{\rm CW}}(t_o-2[\lambda^{-1}]_o)}{\tau_{\rm CW}(t_o)}\right)^{\frac{1}{2}}
\frac{\tau_{\rm CW}(t_o-2[\lambda^{-1}]_o)}{\tau_{\rm CW}(t_o)}e^{\xi(t_o,\lambda)}\nonumber\\
&=(2\lambda)^{-\frac{1}{2}}\sqrt{\pa_x\left(\frac{\tau_{{\rm CW}}(t_o-2[\lambda^{-1}]_o)}{\tau_{\rm CW}(t_o)}e^{\xi(t_o,\lambda)}\right)},
\end{align}
which makes it possible to study CKP as the BKP case. Later in \cite{cheng2014}, it is proved that
CKP tau function $\tau_{{\rm CW}}(t_o)$ related with $\tau_{{\rm KP}}(t)$ in the way below
\begin{align}
\tau_{{\rm CW}}(t_o)=\sqrt{\tau_{{\rm KP}}(\dot{t})}.\label{tau-CKP-KP}
\end{align}
Recently in \cite{Krichever-CMP-2021}, it is showed that if
\begin{align*}
\pa_{t_2}\log\tau_{{\rm KP}}(t)|_{t_{e}=0}=0,
\end{align*}
then $\tau_{{\rm KP}}(\dot{t})$ can be viewed as the CKP tau function inherited from KP hierarchy. Then in \cite{Arthamonov2023}, CKP hierarchy is investigated in the framework of Lagrangian Grassmannian and the corresponding hyper--determinant relations are considered. In addition, C--Toda, multi--component CKP and their solutions have been studied in \cite{chang-hu-li2018,lili2019,Krichever-Zabrodin-LMP-2021,
Zabrodin2023,vandeleur-CKP-2023}.

From above history of CKP tau function, there is still a quite basic question unsolved completely in the bosonic construction of CKP tau function. In fact by \eqref{tau-KP-0} and \eqref{tau-CKP-KP}, one may want to know:\\

\begin{tcolorbox}{\textbf{Question:}}

\textit{Whether $\tau_0(t_o)^{-1}=\langle 0| e^{H(t_o)}g |0\rangle^{-1}$ can be viewed as the CKP tau function $\tau_{{\rm CW}}(t_o)$?}

\end{tcolorbox}
\noindent Note that by \eqref{vandeleur-tau} and \eqref{wuchaozhong-tau}, we can find $u_2$ in CKP Lax operator can be expressed in the way below
\begin{align*}
u_2=2\pa^2_x\log\tau_{{\rm CW}}(t_o)=-2\pa_x^2\log \tau_{0}(t_o),
\end{align*}
which preliminarily hints that $\tau_{{\rm CW}}(t_o)=\tau_{0}(t_o)^{-1}$.  However, the complete solution of the question above needs us to prove\\

\fbox{\textit{
{\it $\langle 0| e^{H(t_o)}g |0\rangle^{-1}$
satisfies CKP bilinear equation \eqref{CKP-bilinintro-psi} by setting $\tau_{{\rm CW}}(t_o)=\langle 0| e^{H(t_o)}g |0\rangle^{-1}$ in \eqref{wuchaozhong-tau}.}
}}\\

\noindent Due to the square root term in \eqref{wuchaozhong-tau}, it is really quite non--trivial to do this. Here in this paper, we will try to solve this question by using CKP Darboux transformation.

\subsection{Main results}
The key of this paper is squared eigenfunction potential (SEP) $\Omega(q,r)$ defined by \cite{Oevelpa1993}
\begin{align}\label{SEP-q-r}
\Omega(q,r)_{t_n}={\rm Res}_{\pa}(\pa^{-1}r(L^n)_{\geq 0}q\pa^{-1}),\
\end{align}
where $q$ and $r$ are the eigenfunction and adjoint eigenfunction respectively, corresponding to the KP Lax operator $L$
\begin{align*}
q_{t_n}=(L^n)_{\geq 0}(q),\ r_{t_n}=-(L^n)^*_{\geq 0}(r).
\end{align*}
Note that $\Omega(q,r)$ can be up to some integral constant and when $n=1$, one can find
$$\Omega(q,r)_x=qr.$$
For CKP case, any adjoint eigenfunction $r$ can also be viewed as eigenfunction $q$ by
$-(L^n)^*_{\geq 0}=(L^n)_{\geq 0}, \ n\in\mathbb{Z}^{>0}_{\rm odd}$. Therefore
CKP SEP \cite{cheng2014} can be defined by two CKP eigenfunction $q_1$
and $q_2$ satisfying $\Omega(q_1,q_2)=\Omega(q_2,q_1)$.
Our first main result is about the changes of CKP tau function by using CKP SEP, which is given by the theorem below.
\begin{theorem}\label{Th:darboux-bili}
Given CKP tau function $\tau_{{\rm CW}}(t_o)$
\begin{align*}
\tau_{{\rm CW}}^{\{1\}}(t_o)=\Omega\big(q(t_o),q(t_o)\big)^{\frac{1}{2}}\tau_{{\rm CW}}(t_o),
\end{align*}
is still a CKP tau function satisfying CKP bilinear equation \eqref{CKP-bilinintro-psi} by \eqref{wuchaozhong-tau}, where $q(t_o)$ is a CKP eigenfunction. That is if set
\begin{align*}
\psi^{\{1\}}(t_o,\lambda)=\sqrt{\frac{1}{2\lambda}\pa_x\left(\frac{\tau_{{\rm CW}}^{\{1\}}(t_o-2[\lambda^{-1}]_o)}{\tau_{{\rm CW}}^{\{1\}}(t_o)}e^{\xi(t_o,\lambda)}\right)^2},
\end{align*}
then
\begin{align*}
{\rm Res}_\lambda\psi^{\{1\}}(t_o,\lambda)\psi^{\{1\}}(t'_o,-\lambda)=0.
\end{align*}
\end{theorem}
In order to prove the question raised at the end of Subsection \ref{subsection:CKP}, we need to express  CKP SEP $\Omega(q_1(t_o),q_2(t_o))$ by free bosons, that is the following proposition.
\begin{proposition}\label{Prop:SEP-vac}
Given $\beta_1,\beta_2 \in V=\bigoplus_{j\in \mathbb{Z}+1/2}\mathbb{C}\phi_j$,
\begin{align*}
\Omega(q_1(t_o),q_2(t_o))=\frac{\langle 0| e^{H(t_o)}\beta_1\beta_2g |0\rangle}{\langle 0| e^{H(t_o)}g |0\rangle},
\end{align*}
where $g\in Sp_\infty$ and
\begin{align}
q_i(t_o)=\frac{2\langle 1| e^{H(t_o)}\beta_ig |0\rangle}{\langle 0| e^{H(t_o)}g |0\rangle},\label{CKP-bosonic-eigen}
\end{align}
where $\langle 1|=\langle 0|H_{1/2}$ with $H_{1/2}$ in \eqref{hexpression}.
\end{proposition}
Based upon the preparation above, now we have the second major result given in the below theorem.
\begin{theorem}\label{Th:CW-van}
If assume that $\beta_i \in V$($1\leq i\leq n$), then
\begin{align*}
\tau_{{\rm CW}}(t_o)=\langle 0| e^{H(t_o)}e^{\frac{{\beta}^2_{n}}{2}}e^{\frac{{\beta}^2_{n-1}}{2}}\cdots
e^{\frac{{\beta}^2_{1}}{2}} |0\rangle^{-1},
\end{align*}
is the CKP tau function satisfying the CKP bilinear equation \eqref{CKP-bilinintro-psi} by \eqref{wuchaozhong-tau}.
\end{theorem}
\begin{corollary}\label{coro:CKP-tau}
If for $g\in Sp_\infty$, $\tau_{\rm CW}(t_o)=\langle 0| e^{H(t_o)}g |0\rangle^{-1}$ is CKP tau function in \eqref{wuchaozhong-tau}, then for $\beta\in V$,
$$\tau_{\rm CW}^{\{1\}}(t_o)=\langle 0| e^{H(t_o)}e^{\frac{\beta^2}{2}}g |0\rangle^{-1},$$
is another CKP tau function,
which is the transformed CKP tau function starting from $\tau_{\rm CW}(t_o)$ under the CKP Darboux transformation
\begin{align*}
T[q]=1-\frac{q}{\Omega(q,q)-1}\pa^{-1}q,
\end{align*}
with $q$ given by \eqref{CKP-bosonic-eigen}. Further given $\beta_1,\beta_2\in V$, denote $q_1$ and $q_2$ by \eqref{CKP-bosonic-eigen}, then
\begin{align*}
\frac{2\langle1|e^{H(t_o)}\beta_2e^{\frac{\beta_{1}^2}{2}}g|0\rangle}
{\langle0|e^{H(t_o)}e^{\frac{\beta_{1}^2}{2}}g|0\rangle}
=T[q_1](q_2).
\end{align*}

\end{corollary}
\begin{remark}
Here we believe that for most $g\in Sp_\infty$, there exist $\beta_i\in V$ ($1\leq i\leq n$) such that
\begin{align}\label{g-e2}
g|0\rangle=e^{\frac{{\beta}^2_{n}}{2}}e^{\frac{{\beta}^2_{n-1}}{2}}\cdots
e^{\frac{{\beta}^2_{1}}{2}}|0\rangle.
\end{align}
For convenience, we denote $Sp'_\infty$ as the subset of $Sp_\infty$ consisting of elements of the form in \eqref{g-e2}. Then the question raised in Subsection \ref{subsection:CKP} is correct in the case of $g\in Sp'_\infty$.
To see the relation $Sp'_\infty$ with $Sp_\infty$, let us assume $\beta_1$ commutes with $\beta_2$, then by $\beta_1\beta_2=(\frac{\beta_1+\beta_2}{2})^2-(\frac{\beta_1-\beta_2}{2})^2$, we can know
\begin{align*}
e^{\beta_1\beta_2}=e^{(\frac{\beta_1+\beta_2}{2})^2}e^{-(\frac{\beta_1-\beta_2}{2})^2}.
\end{align*}
Therefore for $g\in Sp'_\infty$ and $i,j\in \mathbb{Z}+1/2$ with $i\neq -j$,  we can find by Corollary \ref{coro:CKP-tau} that
$e^{a\phi_i\phi_j} g\in Sp'_\infty$.
 we can use Baker--Campbell--Hausdorff formula and the following fact
\begin{align*}
e^{a:\phi_i\phi_{-i}:}e^{b\phi_k\phi_{l}}|0\rangle
=e^{b'\phi_k\phi_{l}}|0\rangle,
\end{align*}
to rewrite
$$e^{\sum_{i,j}{a_{ij}:\phi_i\phi_j:}}|0\rangle
=\prod_{i\neq -j}e^{b_{ij}\phi_i\phi_j}|0\rangle,$$
then due to $\prod_{i\neq -j}e^{b_{ij}\phi_i\phi_j}\in Sp'_\infty$, we can get $\langle 0| e^{H(t_o)}e^{\sum_{i,j}{a_{ij}:\phi_i\phi_j:}}|0\rangle^{-1}$ is the CKP tau function.
Based upon the discussion above, we believe for most $g\in Sp_\infty$, there exists $g'\in Sp'_\infty$ such that $g|0\rangle=g'|0\rangle$, which means $\langle 0| e^{H(t_o)}g |0\rangle^{-1}$ is CKP tau function. Therefore, the question raised in Subsection \ref{subsection:CKP} should be correct for most  $g\in Sp_\infty$.
\end{remark}
Since we have established the bosonic construction of CKP tau function, the next task is to compute vacuum expectation value $\langle 0| e^{H(t_o)}g |0\rangle$ for $g\in Sp'_\infty$, so that we can have explicit examples of CKP tau functions. The computation of vacuum expectation value of bosonic fields is always related with Hafnian determinant \cite{luque}. Note that the square of bosonic field is usually not zero and thus the exponential of bosonic field is very difficult to deal with, which means that computation of vacuum expectation value $\langle 0| e^{H(t_o)}g |0\rangle$ is usually quite complicated. Here we use Corollary \ref{coro:CKP-tau} to view $\langle 0| e^{H(t_o)}g |0\rangle$ as the result of successive applications of CKP Darboux transformations, so that we can compute $\langle 0| e^{H(t_o)}g |0\rangle$ by classical methods of integrable systems.
\begin{theorem}\label{Them:k-darboux}
If assume $\beta_i \in V$($1\leq i\leq n$)
\begin{align*}
\frac{\langle 0| e^{H(t_o)}e^{\frac{\beta_n^2}{2}}e^{\frac{\beta_{n-1}^2}{2}}\cdots e^{\frac{\beta_1^2}{2}}g |0\rangle}{\langle 0| e^{H(t_o)}g |0\rangle}
=(-1)^{-\frac{n}{2}}\begin{vmatrix}
\Omega_{11}-1 &\Omega_{12}&\cdots&\Omega_{1n}\\
\Omega_{21}&\Omega_{22}-1&\cdots&\Omega_{2n}\\
\cdots&\cdots&\cdots&\cdots\\
\Omega_{n1}&\Omega_{n2}&\cdots&\Omega_{nn}-1
\end{vmatrix}^{-\frac{1}{2}},
\end{align*}
where we choose $g\in Sp_\infty$ such that $\langle 0| e^{H(t_o)}g |0\rangle^{-1}$ is the CKP tau function in \eqref{wuchaozhong-tau} and
\begin{align*}
\Omega_{ij}=\frac{\langle 0| e^{H(t_o)}\beta_i\beta_jg |0\rangle}{\langle 0| e^{H(t_o)}g |0\rangle}.
\end{align*}
\end{theorem}

Finally, we consider the 2--reduction of CKP hierarchy. In \cite{chang2013,Xie-arXiv}, it is stated that KdV hierarchy is the 2--reduction of CKP hierarchy. Here we give a proof for this statement. Therefore, we can see that KdV hierarchy is really quite special, which is not only the reduction of BKP hierarchy \cite{Alexandrov2017,Xie-arXiv}, but also the reduction of CKP hierarchy \cite{chang2013,Xie-arXiv}.  Based upon this, we can use free bosons to discuss KdV solutions. We summarize the corresponding results in the theorem below.
\begin{theorem}\label{Th:SEP-vac-2}
Given KdV Lax operator $\mathfrak{L}=\pa^2+u$, if denote $L=\mathfrak{L}^{\frac{1}{2}}$, then $L$ satisfies the CKP hierarchy, that is,
\begin{align*}
L^*=-L,\quad  L_{t_n}=[(L^n)_{\geq 0},L],\quad n=1,3,5,\cdots.\ \
\end{align*}
And if assume $g=\exp(\sum_{i,j\in \mathbb{Z}+1/2} a_{ij}\phi_i\phi_j)\in Sp'_\infty$ satisfying
\begin{align*}
a_{ij}+a_{ji}=a_{i+2,j-2}+a_{j-2,i+2},
\end{align*}
then
\begin{align*}
\tau_{{\rm CW}}(t_o)=\langle 0|e^{H(t_o)}g|0\rangle^{-1},
\end{align*}
will be the solution of KdV hierarchy as CKP hierarchy.
\end{theorem}

\section{Darboux transformations}
In this section, we firstly review some basic facts about the KP Darboux transformations, particularly for the changes of KP tau functions. Then CKP Darboux transformations are understood as a special KP binary Darboux transformations.  Finally, we investigate transformed CKP tau function under CKP Darboux transformations, which is just Theorem \ref{Th:darboux-bili}.

\subsection{KP Darboux transformations}
For KP hierarchy
\begin{align*}
L_{t_n}=[(L^n)_{\geq 0},L],\  n=1,2,\cdots,
\end{align*}
with $L=\pa+\sum_{j=1}^{\infty}u_{j+1}(t)\pa^{-j},$ the Darboux transformation \cite{chau1992,Oevelpa1993,willox1998,yang2022PhyD} means
\begin{align*}
L\rightarrow L^{[1]}=TLT^{-1},
\end{align*}
still satisfies $L^{[1]}_{t_n}=[(L^{[1]})^n_{\geq 0},L^{[1]}],$ where $T$ is some pseudo--differential operator and called Darboux transformation operator.
It is proved that there exist two types of elementary Darboux transformations
\begin{align*}
T_d[q]=q\pa q^{-1}, \ T_i[r]=r^{-1}\pa^{-1} r,
\end{align*}
where $q$ and $r$ are KP eigenfunction and adjoint eigenfunction respectively, satisfying $q_{t_n}=(L^n)_{\geq 0}(q)$ and $r_{t_n}=-(L^n)^*_{\geq 0}(r).$ Under
Darboux transformation $T$, the objects of KP hierarchy will be transformed by the way in the below table.
\begin{center}
\begin{tabular}{lll}
\multicolumn{3}{c}{Table I. Darboux transformations: KP $\rightarrow$ KP}\\
\hline \hline
${\rm KP\ objects}$ &Transformed KP objects & \ \\
\hline
Lax operator& $ L^{[1]}=TLT^{-1}$
 \\
Dressing operator & $ W^{[1]}=TW\pa^{-k}$
 \\
Wave function& $\psi_{{\rm KP}}(t,\lambda)^{[1]}=T(\psi_{{\rm KP}}(t,\lambda))\lambda^{-k}$ \\
Adjoint wave function& $\psi^*_{{\rm KP}}(t,\lambda)^{[1]}=(T^{-1})^*(\psi^*_{{\rm KP}}(t,\lambda))(-\lambda)^{k}$ \\
Eigenfunction& $ q^{[1]}=T(q)$ \\
Adjoint eigenfunction& $r^{[1]}=(T^{-1})^*(r)$ \\
\hline
\end{tabular}
\end{center}
Here $k$ is the largest order of $T$ with respect to $\pa$ and we use $A^{[n]}$ to denote the transformed object $A$ under $n$--step Darboux transformations.

The changes of KP tau functions under Darboux transformations will be more basic in the KP hierarchy. Note that if assume $L^{[1]}=\pa+\sum u^{[1]}_{j+1}\pa^{-j}$, then under $T_d[q]$ or $T_i[r]$,
\begin{itemize}
  \item $T_d[q]$, $$u_{2}^{[1]}=u_2+\pa_x^2\log q,$$
  \item $T_i[r]$, $$u_{2}^{[1]}=u_2+\pa_x^2\log r,$$
\end{itemize}
while $u_{2}=\pa_x^2\log \tau_{{\rm KP}}$ and $u^{[1]}_{2}=\pa_x^2\log \tau^{[1]}_{{\rm KP}}$, so we can find
\begin{align*}
\tau^{[1]}_{{\rm KP}}=q\cdot\tau_{{\rm KP}}\ {\rm or}\ r\cdot\tau_{{\rm KP}}.
\end{align*}
Then in order to show $\tau^{[1]}_{{\rm KP}}$ is KP tau function, we need to further check whether the following bilinear equations are correct, that is
\begin{align}
{\rm Res}_{\lambda}q(t-[\lambda^{-1}])q(t'+[\lambda^{-1}])\tau_{{\rm KP}}(t-[\lambda^{-1}])\tau_{{\rm KP}}(t'+[\lambda^{-1}])e^{\widetilde{\xi}(t-t',\lambda)}=0, \label{KP-Darboux-qtau}\\
{\rm Res}_{\lambda}r(t-[\lambda^{-1}])r(t'+[\lambda^{-1}])
\tau_{{\rm KP}}(t-[\lambda^{-1}])
\tau_{{\rm KP}}(t'+[\lambda^{-1}])e^{\widetilde{\xi}(t-t',\lambda)}=0. \label{KP-Darboux-rtau}
\end{align}
For this, let us review the theory of spectral representation of KP eigenfunctions \cite{aratyn1998,willox1998}.
\begin{lemma}\cite{aratyn1998,loriswillox-1997}
Given KP eigenfunction $q$ and adjoint eigenfunction $r$,
\begin{align}
q(t)&=-{\rm Res}_{\lambda}\psi_{{\rm KP}}(t,\lambda)\Omega(q(t'),\psi_{{\rm KP}}^*(t',\lambda)),\label{KP-SEP-q}\\
r(t)&={\rm Res}_{\lambda}\psi_{{\rm KP}}^*(t,\lambda)\Omega(\psi_{{\rm KP}}(t',\lambda),r(t')),\label{KP-SEP-r}
\end{align}
where $\Omega$ is SEP defined by \eqref{SEP-q-r}. Further
\begin{align}
&{\rm Res}_{\lambda} \Omega(q(t),\psi_{{\rm KP}}^*(t,\lambda))\Omega(\psi_{{\rm KP}}(t',\lambda), r(t'))=\Omega(q(t),r(t))-\Omega(q(t'),r(t')).\label{KP-SEP-res}
\end{align}
\end{lemma}
\begin{proof}
Firstly if set
$I(t,t')={\rm Res}_{\lambda} \psi_{{\rm KP}}(t,\lambda)\Omega(q(t'),\psi_{{\rm KP}}^*(t',\lambda)),$
we can find by \eqref{KP-bili-psi} and \eqref{SEP-q-r} that
$$\pa_{t'_n}I(t,t')=0,$$
thus we can assume $I(t,t')=f(t)$. Further note that
\begin{align}
\Omega(q(t'),\psi_{{\rm KP}}^*(t',\lambda))&=(-q\lambda^{-1}+\mathcal{O}(\lambda^{-2}))e^{-\widetilde{\xi}(t',\lambda)},\label{KP-SEP-q-O}
\end{align}
which is computed by $\pa^l(e^{\widetilde{\xi}(t,\lambda)})=\lambda^le^{\widetilde{\xi}(t,\lambda)}$ for $l\in \mathbb{Z}$,
\begin{align}\label{KP-psi-W}
\psi_{{\rm KP}}(t,\lambda)=(1+w_1\lambda^{-1}+w_2\lambda^{-2}+w_3\lambda^{-3}+\cdots)e^{\widetilde{\xi}(t,\lambda)},
\end{align}
and based upon \eqref{KP-SEP-q-O} and \eqref{KP-psi-W}, we can get
\begin{align*}
I(t,t')=f(t)=-q(t).
\end{align*}
Similarly we can prove \eqref{KP-SEP-r}.

Since $\Omega(q,r)=\int_{-\infty}^x qr dx$, one can find
\begin{align}
{\rm Res}_\lambda \Omega(q(t),\psi_{{\rm KP}}^*(t,\lambda))\Omega(\psi_{{\rm KP}}(t',\lambda), r(t'))&=\Omega\Big(q(t),{\rm Res}_{\lambda}\psi_{{\rm KP}}^*(t,\lambda)\Omega(\psi_{{\rm KP}}(t',\lambda),r(t'))\Big)+C(t')\nonumber\\
&=\Omega(q(t),r(t))+C(t'),\label{KP-SEP-CT}
\end{align}
where one should remember that $\Omega(q,r)$ can be determined up to some integral constant. If set $t'=t$ in \eqref{KP-SEP-CT} and consider \eqref{KP-SEP-q-O}, we can get
\begin{align}
\Omega(\psi_{\rm KP}(t,\lambda),r(t))=(r\lambda^{-1}+\mathcal{O}(\lambda^{-2}))e^{\widetilde{\xi}(t,\lambda)},\label{KP-SEP-r-O}
\end{align}
then
$$C(t)+\Omega(q(t),r(t))=0,$$
which leads to \eqref{KP-SEP-res}.
\end{proof}
If set $t-t'=[\mu^{-1}]$ in \eqref{KP-SEP-q}--\eqref{KP-SEP-res} and use \eqref{KP-SEP-q-O}\eqref{KP-SEP-r-O}, then we can get the corollary below.
\begin{corollary}\cite{aratyn1998,willox1998}
Given KP eigenfunction $q$ and adjoint KP eigenfunction $r$,
\begin{align}
&\Omega(q(t),\psi_{{\rm KP}}^*(t,\lambda))=-\lambda^{-1}q(t+[\lambda^{-1}])\psi_{{\rm KP}}^*(t,\lambda),\label{KP-SEP-psi-q}\\
&\Omega(\psi_{{\rm KP}}(t,\lambda),r(t))=\lambda^{-1}\psi_{{\rm KP}}(t,\lambda)r(t-[\lambda^{-1}]),\nonumber\\
&\Omega(q(t-[\lambda^{-1}]),r(t-[\lambda^{-1}]))-\Omega(q(t),r(t))=\lambda^{-1}q(t)r(t-[\lambda^{-1}]).\nonumber
\end{align}
\end{corollary}
\noindent So by \eqref{KP-SEP-q}\eqref{KP-SEP-r}, we can find
\begin{align}
&q(t)={\rm Res}_{\lambda}\lambda^{-1}\psi_{{\rm KP}}(t,\lambda)\psi_{{\rm KP}}^*(t',\lambda)q(t'+[\lambda^{-1}]),\label{KP-q-res}\\
&r(t)={\rm Res}_{\lambda}\lambda^{-1}\psi_{{\rm KP}}^*(t,\lambda)\psi_{{\rm KP}}(t',\lambda)r(t'-[\lambda^{-1}]).\label{KP-r-res}
\end{align}
Next let $t-t'=[\lambda^{-1}]+[\mu^{-1}]$ in \eqref{KP-q-res}\eqref{KP-r-res} and we have the corollary below.
\begin{corollary}\cite{willox1998}
Given KP eigenfunction $q$ and adjoint KP eigenfunction $r$,
\begin{align}
&\lambda q(t-[\lambda^{-1}])\psi_{{\rm KP}}(t,\lambda)=\mu \Big(\psi_{{\rm KP}}(t,\lambda)q(t-[\mu^{-1}])-\psi_{{\rm KP}}(t-[\mu^{-1}],\lambda)q(t)\Big),  \label{KP-SEP-lambda-q}\\
&\lambda r(t+[\lambda^{-1}])\psi_{{\rm KP}}^*(t,\lambda)=\mu \Big(\psi_{{\rm KP}}^*(t,\lambda)r(t+[\mu^{-1}])-\psi_{{\rm KP}}^*(t+[\mu^{-1}],\lambda)r(t)\Big).\label{KP-SEP-lambda-r}
\end{align}
\end{corollary}

\begin{proposition}\label{prop:KP-tau-Darboux}\cite{willox1998,yang2022PhyD}
If $\tau_{{\rm KP}}$ is KP tau function, then for KP eigenfunction $q$ and adjoint KP eigenfunction $r$,
$$\tau_{{\rm KP}}\rightarrow\tau^{[1]}_{{\rm KP}}=q\cdot\tau_{{\rm KP}}\ {\rm or}\ r\cdot\tau_{{\rm KP}},$$
 will be another KP tau function satisfying KP bilinear equation \eqref{KP-bilinintro-tau}.
\end{proposition}
\begin{proof}
Just as we stated before, we only need to prove \eqref{KP-Darboux-qtau} and \eqref{KP-Darboux-rtau}. For \eqref{KP-Darboux-qtau}, it can be proved by \eqref{KP-SEP-q} \eqref{KP-SEP-psi-q} \eqref{KP-SEP-lambda-q}, while according \eqref{KP-SEP-r} \eqref{KP-SEP-psi-q} \eqref{KP-SEP-lambda-r}, we can prove \eqref{KP-Darboux-rtau}.
\end{proof}
For elementary Darboux transformations $T_d$ and $T_i$, the following diagram can commute
$$\xymatrix{&&L^{[1]}\ar[drr]^{T_\beta^{[1]}}&&\\
L^{[0]}\ar[urr]^{T_\alpha}\ar[drr]_{T_\beta}& &&& L^{[2]}\\
&&L^{[1]}\ar[urr]_{T_\alpha^{[1]}}&&}$$
where $\alpha,\beta\in \{d,i\}$, that is,
\begin{align*}
 T_i[r^{[1]}]T_d[q]=T_d[q^{[1]}]T_i[r],\  T_d[q^{[1]}_1]T_d[q_2]=T_d[q^{[1]}_2]T_d[q_1],\  T_i[r^{[1]}_1]T_i[r_2]=T_i[r^{[1]}_2]T_i[r_1].
\end{align*}
Due to the commutativity of elementary Darboux transformations, we can consider the following chain of Darboux transformations
\begin{eqnarray*}
&&L\xrightarrow{T_d[q_1]}L^{[1]}\xrightarrow{T_d[q_2^{[1]}]}L^{[2]}
  \rightarrow\cdots\rightarrow L^{[n-1]}\xrightarrow{T_d[q_n^{[n-1]}]}L^{[n]}\\
&&\xrightarrow{T_i[r_1^{[n]}]}L^{[n+1]}\xrightarrow{T_i[r_2^{[n+1]}]}\cdots\rightarrow L^{[n+k-1]}\xrightarrow{T_i[r_k^{[n+k-1]}]}L^{[n+k]},
\end{eqnarray*}
where $q_i$ and $r_j$ are KP eigenfunctions and adjoint eigenfunctions. If denote
\begin{align*}
T^{[n,k]}=T_i(r_k^{[n+k-1]})\cdots
T_i(r_1^{[n]})T_d(q_n^{[n-1]})\cdots T_d(q_1),
\end{align*}
then $L^{[n+k]}=T^{[n,k]}L(T^{[n,k]})^{-1}$.

Here we are more interested in the case of $n=k=1$, which is just the binary Darboux transformation $T(q,r)$, that is
\begin{align*}
T(q,r)=T_d[q^{[1]}]T_i[r]=1-\Omega^{-1}(q,r)q\pa^{-1}r.
\end{align*}
The inverse of $T(q,r)$ is $T_i[r]^{-1}T_d[q^{[1]}]^{-1}$, i.e. $$T(q,r)^{-1}=1+q\pa^{-1}r\Omega(q,r)^{-1}.$$
We use the symbol $A^{\{n\}}$ to denote the transformed object A under $n$--step binary Darboux transformations.

\begin{corollary}
If $\tau_{{\rm KP}}(t)$ is the KP tau function, then under 1--step binary Darboux transformation $T(q,r)$,
\begin{align*}
\tau^{\{1\}}_{{\rm KP}}(t)=\Omega(q,r)\tau_{{\rm KP}}(t),
\end{align*}
which satisfies
\begin{align*}
{\rm Res}_\lambda\tau^{\{1\}}_{{\rm KP}}(t-[\lambda^{-1}])\tau^{\{1\}}_{{\rm KP}}(t'+[\lambda^{-1}])e^{\widetilde{\xi}(t-t',\lambda)}=0.
\end{align*}
\end{corollary}
\begin{proof}
Firstly under $T_{i}(r)$, $q^{[1]}=r^{-1}\Omega(q,r)$. Then
$$\tau^{\{1\}}_{{\rm KP}}(t)=\Omega(q,r)\tau_{{\rm KP}}(t)=q^{[1]}r\cdot\tau_{{\rm KP}}(t).$$
So by Proposition \ref{prop:KP-tau-Darboux}, we can prove this corollary.
\end{proof}

Given KP eigenfunction $q'$ and adjoint eigenfunction $r'$, under binary Darboux transformation $T(q,r)$, $\Omega(q',r')$ will be transformed into
\begin{align*}
\Omega(q'^{\{1\}},r'^{\{1\}})=\Omega(q',r')-\Omega(q',r)\Omega(q,r')\Omega(q,r)^{-1}.
\end{align*}
\begin{corollary}\label{coro:KP-Darboux-n}\cite{he2002}
Under $n$--step binary Darboux transformation
\begin{align*}
T^{\{n\}}=T(q_n^{\{n-1\}},r_n^{\{n-1\}})
T(q_n^{\{n-2\}},r_n^{\{n-2\}})
\cdots T(q_1,r_1),
\end{align*}
where $q_i$ and $r_j$ are KP eigenfunctions and adjoint eigenfunctions,
the KP tau function $\tau_{{\rm KP}}$ will be transformed into
\begin{align}\label{KP-binary-tau}
\tau^{\{n\}}_{{\rm KP}}
=\Omega_{nn}^{\{n-1\}}\Omega_{n-1, n-1}^{\{n-2\}}\cdots\Omega_{22}^{\{1\}}\Omega_{11}\cdot\tau_{{\rm KP}}=
\left|\begin{array}{ccc}
\Omega_{11} & \cdots &\Omega_{1n}\\
\cdots & \cdots &\cdots \\
\Omega_{n1}& \cdots &\Omega_{nn}
\end{array}
\right|\cdot\tau_{{\rm KP}},
\end{align}
where $\Omega_{ij}^{\{k\}}=\Omega(q_i^{\{k\}},r_j^{\{k\}})$ and $\Omega_{ij}=\Omega(q_i,r_j)$.
\end{corollary}
\begin{proof}
Firstly, it is obviously that $\tau^{\{n\}}_{{\rm KP}}
=\Omega_{nn}^{\{n-1\}}\Omega_{n-1, n-1}^{\{n-2\}}\cdots\Omega_{22}^{\{1\}}\Omega_{11}\cdot\tau_{{\rm KP}}$ step by step. Then let us assume that \eqref{KP-binary-tau} is correct for $n-1$. Note that $\tau_{{\rm KP}}^{\{n\}}$ can be viewed as the transformed KP tau function under $(n-1)$--step binary Darboux transformations, starting from $\tau_{{\rm KP}}^{\{1\}}=\Omega_{11}\tau_{{\rm KP}}$, thus
\begin{align*}
\tau^{\{n\}}_{{\rm KP}}=
\left|\begin{array}{ccc}
\Omega^{\{1\}}_{22} & \cdots &\Omega^{\{1\}}_{2n}\\
\cdots & \cdots &\cdots \\
\Omega^{\{1\}}_{n2}& \cdots &\Omega^{\{1\}}_{nn}
\end{array}
\right|\cdot\tau_{{\rm KP}}^{\{1\}}=\left|\begin{array}{ccc}
\Omega^{\{1\}}_{22} & \cdots &\Omega^{\{1\}}_{2n}\\
\cdots & \cdots &\cdots \\
\Omega^{\{1\}}_{n2}& \cdots &\Omega^{\{1\}}_{nn}
\end{array}
\right|\cdot\Omega_{11}\tau_{{\rm KP}}.
\end{align*}
Next for the determinant below,
\begin{align*}
\left|\begin{array}{cccc}
\Omega_{11}&\Omega_{12} & \cdots &\Omega_{1n}\\
0&\Omega^{\{1\}}_{22} & \cdots &\Omega^{\{1\}}_{2n}\\
\cdots&\cdots & \cdots &\cdots \\
0&\Omega^{\{1\}}_{n2}& \cdots &\Omega^{\{1\}}_{nn}
\end{array}
\right|=\left|\begin{array}{ccc}
\Omega^{\{1\}}_{22} & \cdots &\Omega^{\{1\}}_{2n}\\
\cdots & \cdots &\cdots \\
\Omega^{\{1\}}_{n2}& \cdots &\Omega^{\{1\}}_{nn}
\end{array}
\right|\cdot\Omega_{11},
\end{align*}
if add the multiplication of the first row by $\Omega_{i1}\Omega_{11}^{-1}$ to the $i$--th row for $2\leq i\leq n$ and notice that
\begin{align*}
\Omega^{\{1\}}_{ij}=\Omega_{ij}-\Omega_{i1}\Omega_{1j}\Omega_{11}^{-1},
\end{align*}
we can finally obtain
\begin{align*}
\tau^{\{n\}}_{{\rm KP}}=
\left|\begin{array}{ccc}
\Omega_{11} & \cdots &\Omega_{1n}\\
\cdots & \cdots &\cdots \\
\Omega_{n1}& \cdots &\Omega_{nn}
\end{array}
\right|\cdot\tau_{{\rm KP}}.
\end{align*}
\end{proof}

\subsection{CKP Darboux transformations}\label{CKP:Darboux}
Recall that the CKP hierarchy \cite{Date1981JPAJ,Zabrodin2021} is defined by
\begin{align*}
L^*=-L,
\end{align*}
and
\begin{align*}
L_{t_n}=[(L^n)_{\geq 0},L],\  n=1,3,5,\cdots,
\end{align*}
with Lax operator $$L=\pa+\sum_{j=1}^{\infty}u_{j+1}(t_o)\pa^{-j}.$$
In CKP case, all the adjoint eigenfunctions are also eigenfunctions and vice versa. Different from the KP case, the CKP Darboux transformation operator $T$ should satisfy the following two conditions \cite{HeJMP2007}
\begin{itemize}
  \item $(L^{[1]})^*=-L^{[1]},$
  \item $L^{[1]}_{t_n}=[(L_{\geq 0}^{[1]})^n,L^{[1]}],$
\end{itemize}
where $L^{[1]}=TLT^{-1}$. The elementary CKP Darboux transformation is given by
$$T[q]\triangleq T(q,q)=T_d[q^{[1]}]T_i[q]=1-\Omega^{-1}(q,q)q\pa^{-1}q,$$
where $q=q(t_o)$ is CKP eigenfunction.

Let us recall that the CKP tau function $\tau_{{\rm CW}}(t_o)$ is related with the CKP wave function $\psi(t_o,\lambda)$ in the way below \cite{chang2013}.
\begin{eqnarray}
\psi(t_o,\lambda)=\left(1+\frac{1}{\lambda}\pa_x{\rm log}\frac{\tau_{{\rm CW}}(t_o-2[\lambda^{-1}]_o)}{\tau_{\rm CW}(t_o)}\right)^{\frac{1}{2}}
\frac{\tau_{\rm CW}(t_o-2[\lambda^{-1}]_o)}{\tau_{\rm CW}(t_o)}e^{\xi(t_o,\lambda)}.\label{ckpwavetau}
\end{eqnarray}
For convenience, we will represent \eqref{ckpwavetau} as
\begin{align*}
\psi(t_o,\lambda)=G(t_o,\lambda)\cdot\varphi(t_o,\lambda)=(2\lambda)^{-\frac{1}{2}}\sqrt{\pa_x\varphi^2(t_o,\lambda)},
\end{align*}
where
\begin{align*}
&\varphi(t_o,\lambda)=\frac{\tau_{{\rm CW}}(t_o-2[\lambda^{-1}]_o)}{\tau_{\rm CW}(t_o)}e^{\xi(t_o,\lambda)},\ \ G(t,\lambda)=\left(1+\frac{1}{\lambda}\pa_x{\rm log}\frac{\tau_{{\rm CW}}(t_o-2[\lambda^{-1}]_o)}{\tau_{\rm CW}(t_o)}\right)^{\frac{1}{2}}.
\end{align*}
Here we will investigate the change of CKP tau function $\tau_{{\rm CW}}$ under the Darboux transformation $T[q]$.

Firstly by comparing the coefficients of $\pa^{-1}$ in
\begin{align*}
W^{\{1\}}=T[q]W,
\end{align*}
for CKP dressing operator $W=1+\sum_{j=1}^\infty w_j \partial^{-j}$ and using \eqref{ckpwavetau} with $\psi(t_0,\lambda)=W(e^{\xi(t_0,\lambda)})$, we can obtain
\begin{align*}
2\pa_x\log\tau^{\{1\}}_{{\rm CW}}=\pa_x\log\Omega(q,q)+2\pa_x\log\tau_{{\rm CW}},
\end{align*}
which implies that
\begin{align*}
\tau_{{\rm CW}}^{\{1\}}=\Omega^{\frac{1}{2}}(q,q)\tau_{{\rm CW}}.
\end{align*}
So we need to prove $\Omega^{\frac{1}{2}}(q,q)\tau_{{\rm CW}}$ still satisfies CKP bilinear equation, that is Theorem \ref{Th:darboux-bili}.

To do this, let us review properties of CKP SEP \cite{cheng2014}. Firstly given CKP eigenfunction $q(t_o)$,
\begin{align}\label{CKP-SEP-q}
q(t_o)=-{\rm Res}_{\lambda}\psi(t_o,\lambda)\Omega(q(t'_o),\psi(t'_o,-\lambda)),
\end{align}
which is just the special case of \eqref{KP-SEP-q} by noting the CKP adjoint wave function $\psi^*(t_o,\lambda)=\psi(t_o,-\lambda)$.
Similar to \eqref{KP-SEP-psi-q}, if set $t'_o-t_o=2[\mu^{-1}]_o$ in \eqref{CKP-SEP-q}, one can get
\begin{align}\label{SEP-q-psi}
\Omega(q(t_o),\psi(t_o,\lambda))=\frac{\psi(t_o,\lambda)}{2\lambda G^2(t_o,\lambda)}\Big(q(t_o)+q(t_o-2[\lambda^{-1}]_o)\Big).
\end{align}
\begin{remark}
Given two CKP eigenfunctions $q_1$ and $q_2$, one can find \cite{cheng2014}
\begin{align*}
\Omega(q_1,q_2)=\Omega(q_2,q_1).
\end{align*}
Therefore we expect
\begin{align}\label{CKP-SEP-psi-comm}
\Omega(\psi(t_o,\lambda),\psi(t_o,\mu))=\Omega(\psi(t_o,\mu),\psi(t_o,\lambda)),
\end{align}
but it is difficult to see this directly
from the corresponding expressions in \eqref{SEP-q-psi}.
In fact by equation (2.25) in \cite{Krichever-CMP-2021} (there is a small typo in the sign of b), we can know
\begin{align}\label{CKP-Zab}
\frac{1}{a+b}\left(af(t_o,a)\frac{w(t_o-2[b^{-1}]_o,a)}{w(t_o,a)}
+bf(t_o,b)\frac{w(t_o-2[a^{-1}]_o,b)}{w(t_o,b)}\right)
=\frac{af(t_o,a)-bf(t_o,b)}{a-b},
\end{align}
where $\psi(t_o,\lambda)=w(t_o,\lambda)e^{\xi(t_o,\lambda)}$, $w(t_o,\lambda)=1+\sum_{i=1}^{\infty}w_i(t_o)\lambda^{-i}$ and
\begin{align*}
f(t_o,\lambda)=G^2(t_o,\lambda)=1+\frac{1}{2\lambda}\Big(w_1(t_o)-w_1(t_o-2[\lambda^{-1}]_o)\Big).
\end{align*}
Then by \eqref{CKP-Zab}, we can finally find
\begin{align*}
\frac{\psi(t_o,\lambda)}{2\lambda f(t_o,\lambda)}\Big(\psi(t_o,\mu)+\psi(t_o-2[\lambda^{-1}]_o,\mu)\Big)
-\frac{\psi(t_o,\mu)}{2\mu f(t_o,\mu)}\Big(\psi(t_o,\lambda)+\psi(t_o-2[\mu^{-1}]_o,\lambda)\Big)=\delta(\lambda,-\mu),
\end{align*}
which shows that \eqref{CKP-SEP-psi-comm} holds up to constant $\delta(\lambda,-\mu)=\sum_{k\in \mathbb{Z}}\lambda^{k-1}\mu^{-k}$.
\end{remark}
\begin{lemma}\label{lemma:SEP-t-2z}
Given two CKP eigenfunctions $q_1$ and $q_2$,
\begin{align*}
\Omega_{12}(t_o-2[\lambda^{-1}])-\Omega_{12}(t_o)=-\frac{2\lambda}{\varphi^2(t_o,\lambda)}\Omega\big(q_1(t_o),\psi(t_o,\lambda)\big)\cdot\Omega\big(q_2(t_o),\psi(t_o,\lambda)\big)
,
\end{align*}
where $\Omega_{12}(t_o)=\Omega(q_1(t_o),q_2(t_o))$.
\end{lemma}
\begin{proof}
Firstly by \eqref{KP-SEP-res} and CKP adjoint wave function $\psi^*(t_o,\lambda)=\psi(t_o,-\lambda)$,
\begin{align}\label{CKP-SEP-res}
{\rm Res}_\lambda \Omega(q_1(t_o),\psi(t_o,\lambda))\Omega(q_2(t_o'),\psi(t_o',-\lambda))=\Omega_{12}(t'_o)-\Omega_{12}(t_o).
\end{align}
Next if assume $\Omega(q_i(t_o),\psi(t_o,\lambda))=K_i(t_o,\lambda)e^{\xi(t_o,\lambda)}$, then by \eqref{KP-SEP-q-O} and \eqref{SEP-q-psi},
\begin{align*}
K_i(t_o,\lambda)=\frac{w(t_o,\lambda)}{2\lambda G^2(t_o,\lambda)}\Big(q_i(t_o)+q_i(t_o-2[\lambda^{-1}]_o)\Big)=q_i(t_o)\lambda^{-1}+\mathcal{O}(\lambda^{-2}).
\end{align*}
By setting $t_o-t'_o=2[\mu^{-1}]_o$ in \eqref{CKP-SEP-res} and computing the residue at $\lambda=\mu$, we can find
\begin{align*}
&\Omega_{12}(t_o)-\Omega_{12}(t_o-2[\mu^{-1}]_o)=2\mu K_1(t_o,\mu)K_2(t_o-2[\mu^{-1}]_o,-\mu)\\
&=-\frac{\big(q_1(t_o)+q_1(t_o-2[\mu^{-1}]_o)\big)\big(q_2(t_o)+q_2(t_o-2[\mu^{-1}]_o)\big)}{2\mu G^2(t_o,\mu)}\\
&=-\frac{2\mu\Omega(q_1(t_o),\psi(t_o,\mu))\Omega(q_2(t_o),\psi(t_o,\mu))}{\varphi^2(t_o,\mu)}.
\end{align*}
Here we have used $w(t_o,\mu)w(t_o-2[\mu^{-1}]_o,-\mu)=f(t_o,\mu)=G^2(t_o,\mu)$, which is derived by setting $t_o'=t_o-2[\mu^{-1}]_o$ in CKP bilinear equation \eqref{CKP-bilinintro-psi} (one can refer to \cite{Krichever-CMP-2021}).
\end{proof}

\subsection{Proof of Theorem \ref{Th:darboux-bili}}

Firstly by Lemma \ref{lemma:SEP-t-2z}, one can know that
\begin{align}\label{CKP-Darboux-psi}
\psi^{\{1\}}(t_o,\lambda)=\psi(t_o,\lambda)-\frac{q(t_o)\Omega\big(q(t_o),\psi(t_o,\lambda)\big)}{\Omega(t_o)},
\end{align}
where $\Omega(t_o)=\Omega(q(t_o),q(t_o))$. Then by \eqref{CKP-bilinintro-psi}\eqref{CKP-SEP-q} and \eqref{CKP-SEP-res},
\begin{align*}
&{\rm Res}_\lambda\psi^{\{1\}}(t_o,\lambda)\psi^{\{1\}}(t'_o,-\lambda)\\
=&{\rm Res}_\lambda\psi(t_o,\lambda)\psi(t'_o,-\lambda)
+\frac{q(t_o)q(t'_o)}{\Omega(t_o)\Omega(t'_o)}{\rm Res}_\lambda\Omega(q(t_o),\psi(t_o,\lambda))\Omega(q(t'_o),\psi(t'_o,-\lambda))\\
-&\frac{q(t_o)}{\Omega(t_o)}{\rm Res}_\lambda\Omega(q(t_o),\psi(t_o,\lambda))\psi(t'_o,-\lambda)
-\frac{q(t'_o)}{\Omega(t'_o)}{\rm Res}_\lambda\psi(t_o,\lambda)\Omega(q(t'_o),\psi(t'_o,-\lambda))=0,
\end{align*}
which finishes the proof of Theorem \ref{Th:darboux-bili}.

\begin{remark}\label{remark:taucw-Darboux}
By \eqref{CKP-Darboux-psi}, we can know that $\tau_{{\rm CW}}(t_o)\longrightarrow \tau^{\{1\}}_{{\rm CW}}(t_o)=\Omega(t_o)^{\frac{1}{2}}\tau_{{\rm CW}}(t_o)$ is corresponding to CKP Darboux transformation $T[q]=1-\Omega(q,q)^{-1}q\partial^{-1}q$, that is
\begin{align*}
\psi(t_o,\lambda)\longrightarrow\psi^{\{1\}}(t_o,\lambda)=T[q](\psi(t_o,\lambda)).
\end{align*}
If we use $a\Omega(q,q)+b,\ (a\neq0)$ instead of $\Omega(q,q)$, then
\begin{align*}
\tau_{{\rm CW}}(t_o)&\longrightarrow \tau^{\{1\}}_{{\rm CW}}(t_o)=(a\Omega(t_o)+b)^{\frac{1}{2}}\tau_{{\rm CW}}(t_o),\\
\psi(t_o,\lambda)&\longrightarrow\psi^{\{1\}}(t_o,\lambda)
=\psi(t_o,\lambda)-\frac{a q(t_o)\Omega(q(t_o),\psi(t_o,\lambda))}{a\Omega(t_o)+b}.
\end{align*}
Similar to the proof of Theorem \ref{Th:darboux-bili}, this $\psi^{\{1\}}(t_o,\lambda)$ still satisfies the CKP bilinear equation. In fact, $\Omega(t_o)$ can be up to some integral constant and $\tau_{{\rm CW}}(t_o)$ can be up to the multiplication of $c \cdot {\rm exp}\left(\sum_{n=0}^{\infty}t_{2n+1}a_{2n+1}\right)$. In this case,
\begin{align*}
T[q]=1-\frac{a q}{a\Omega(t_o)+b}\pa^{-1}q.
\end{align*}
\end{remark}
\begin{corollary}
Consider the following chain of CKP Darboux transformations
\begin{eqnarray*}
L\xrightarrow{T[q_1]}L^{\{1\}}\xrightarrow{T[q_2^{\{1\}}]}L^{\{2\}}
  \rightarrow\cdots\rightarrow L^{\{n-1\}}\xrightarrow{T(q_n^{\{n-1\}})}L^{\{n\}},
\end{eqnarray*}
where $q_i$ are CKP eigenfunctions corresponding to CKP Lax operator $L$, then
\begin{align*}
\tau_{{\rm CW}}^{\{n\}}=\sqrt{\left|\begin{array}{ccc}
\Omega_{11} & \cdots &\Omega_{1n}\\
\cdots & \cdots &\cdots \\
\Omega_{n1}& \cdots &\Omega_{nn}
\end{array}
\right|}\tau_{{\rm CW}}(t_o),
\end{align*}
where $\Omega_{ij}=\Omega(q_i(t_o),q_j(t_o))$.
\end{corollary}
\begin{proof}
By Theorem \ref{Th:darboux-bili}, we can find
\begin{align*}
\tau_{{\rm CW}}^{\{n\}}(t_o)
=\sqrt{\Omega^{\{n-1\}}_{nn}\cdots\Omega^{\{1\}}_{22}\Omega_{11}}\cdot\tau_{{\rm CW}}(t_o).
\end{align*}
Then according to Corollary \ref{coro:KP-Darboux-n}, we can prove this corollary.
\end{proof}

\section{Bosonic construction of CKP tau function}
In this section, we firstly review some basic facts on free boson and its bosonizations. Then by using the representation of Lie algebra $c_\infty$ in bosonic Fock space, the corresponding bosonic CKP hierarchy is obtained, which can be viewed as the sub--hierarchy of super BKP hierarchy. Next by using the bosonization of free boson, the bilinear equation of CKP hierarchy is derived. After above preparation, we prove the major theorem, that is Theorem \ref{Prop:SEP-vac}. Finally by using results of successive applications of CKP Darboux transformations, we compute the corresponding vacuum expectation value of bosonic fields, that is Theorem \ref{Them:k-darboux}.

\subsection{Free boson and its bosonization}\label{subsec:boson}
Firstly free bosons $\phi_i(i\in \mathbb{Z}+1/2)$ \cite{vandeleur2023,vandeleur2012,yang2021jmp,kac1989} are defined by
\begin{eqnarray}
\phi_i\phi_j-\phi_j\phi_i=(-1)^{j-\frac{1}{2}}\delta_{i,-j}.\label{phiiphij}
\end{eqnarray}
Then the corresponding Fock space $\mathcal{F}$ and its dual $\mathcal{F}^*$ are given by the vectors spaces below,
\begin{align*}
&\mathcal{F}={\rm span}\{(\phi_{j_1})^{m_1} (\phi_{j_2})^{m_2} \cdots (\phi_{j_n})^{m_n}|0\rangle\ |j_1>j_2>\cdots>j_n>0,\  m_i\geq 0\},\\
&\mathcal{F}^*={\rm span}\{\langle 0|(\phi_{-j_n})^{m_n}\cdots(\phi_{-j_2})^{m_2}(\phi_{-j_1})^{m_1} |j_1>j_2>\cdots>j_n>0,\  m_i\geq 0\},
\end{align*}
with vacuum vectors $|0\rangle$ and $\langle 0|$ defined by
\begin{eqnarray*}
\phi_{-j}|0\rangle=0,\quad\langle 0|\phi_{j}=0,\quad j>0.
\end{eqnarray*}

Define the pairing of $a|0\rangle \in \mathcal{F}$ and
$\langle 0|b\in \mathcal{F}^*$ to be $\langle 0|ba|0\rangle$, which is determined by $\langle 0|0\rangle=1$ and \eqref{phiiphij}. Note that
\begin{align*}
\langle 0|\phi_i\phi_j|0\rangle=(-1)^{j-\frac{1}{2}}\delta_{i,-j}\theta(j>0),
\end{align*}
where $\theta(P)$ is the Boolean characteristic function of $P$, that is $\theta(P)=1$ if $P$ true, while $\theta(P)=0$ if $P$ false. If introduce the generating function of free bosons in the way below,
\begin{align*}
\phi(\lambda)=\sum_{j\in \mathbb{Z}+\frac{1}{2}}\phi_j\lambda^{j-\frac{1}{2}},
\end{align*}
then
\begin{align*}
\langle 0|\phi(p)\phi(q)|0\rangle=\frac{1}{p+q}.
\end{align*}
Next we will consider the general case $\langle 0|\beta_n\cdots\beta_2\beta_1|0\rangle$, where $\beta_i\in V=\bigoplus\limits_{j\in \mathbb{Z}+1/2}\mathbb{C}\phi_j$.

For this, let us introduce an operator $S$ on $\mathcal{F}\otimes \mathcal{F}$ defined by
\begin{align*}
S=\sum_{k\in \mathbb{Z}+\frac{1}{2}}(-1)^{k+\frac{1}{2}}\phi_k\otimes\phi_{-k}
={\rm Res}_\lambda\phi(\lambda)\otimes\phi(-\lambda).
\end{align*}
Then by \eqref{phiiphij}, we have the lemma below.
\begin{lemma}\label{lemma:S-beta}\cite{yang2022PhyD}
Given $\beta_i\in V=\bigoplus_{j\in \mathbb{Z}+1/2}\mathbb{C}\phi_j$,
\begin{align*}
&S(1\otimes \beta_1\beta_{2}\cdots\beta_n)
=-\sum_{l=1}^n\beta_l\otimes\beta_1\beta_{2}\cdots\widehat{\beta}_l\cdots\beta_n
+(1\otimes\beta_1\beta_{2}\cdots\beta_n)S,\\
&S(\beta_1\beta_{2}\cdots\beta_n\otimes1)
=\sum_{l=1}^n\beta_1\beta_{2}\cdots\widehat{\beta}_l\cdots\beta_n\otimes\beta_l
+(\beta_1\beta_{2}\cdots\beta_n\otimes 1)S,
\end{align*}
where
$$\beta_1\beta_{2}\cdots\widehat{\beta}_l\cdots\beta_n=\beta_1\beta_{2}\cdots\beta_{l+1}\beta_{l-1}\cdots\beta_n.$$
\end{lemma}
By this lemma, one can find
\begin{align*}
&S(\beta_{2}\beta_{3}\cdots\beta_{2n}\otimes\beta_1)
=S(1\otimes\beta_1)(\beta_2\beta_3\cdots\beta_{2n}\otimes 1)\\
&=(1\otimes\beta_1)S(\beta_2\beta_3\cdots\beta_{2n}\otimes 1)
-\beta_1\beta_2\cdots\beta_{2n}\otimes 1\\
&=\sum_{l=2}^{2n}(\beta_{2}\beta_{3}\cdots\widehat{\beta}_{l}\cdots\beta_{2n}\otimes\beta_{1}\beta_l)
+(\beta_{2}\beta_{3}\cdots\beta_{2n}\otimes\beta_{1})S-(\beta_{1}\beta_{2}\cdots\beta_{2n}\otimes1).
\end{align*}
If apply $\langle 0|\otimes\langle0|$ and $|0\rangle\otimes|0\rangle$ to the above relation,
\begin{align}
\langle0|\beta_1\beta_{2}\cdots\beta_{2n}|0\rangle
=\sum_{l=2}^{2n}\langle0|\beta_{1}\beta_l|0\rangle \langle0|\beta_{2}\beta_{3}\cdots\widehat{\beta}_{l}
\cdots\beta_{2n}|0\rangle,
\label{betaexpansion}
\end{align}
where we have used
\begin{align*}
(\langle 0|\otimes\langle0|)S=0,\ \  S(|0\rangle\otimes|0\rangle)=0.
\end{align*}

On the other hand, let us review some knowledge on Hafnian determinant \cite{luque,hirai}. For a $2n\times2n$  symmetric matrix $A=(a_{ij})$, the Hafnian of A is defined by
\begin{align}
{\rm Haf} A=\frac{1}{n!2^n}\sum_{\pi \in \mathcal{S}_{2n}}\prod_{j=1}^na_{\pi(2j-1),\pi(2j)},\label{hafdefinition}
\end{align}
where $\mathcal{S}_{2n}$ is the symmetric group on $2n$
 elements. If denote $A[i,j]$ to be the matrix obtained from $A$ by removing $i$--th row, $j$--th row, $i$--th column and $j$--th column, then \cite{hirai}
\begin{align}\label{Haf-2}
{\rm Haf} A=\sum_{j\neq i}{\rm Haf} A[i,j],\quad \text{for some fixed $i$.}
\end{align}
Thus if assume $2n\times2n$ symmetric matrix $M=(m_{ij})^{2n}_{i,j=1}$ with
\begin{align*}
m_{ij}=\langle 0|\beta_i\beta_j|0\rangle,\quad \beta_i\in V=\bigoplus_{j\in \mathbb{Z}+1/2}\mathbb{C}\phi_j,\quad \beta_i\beta_j=\beta_j\beta_i,
\end{align*}
then by \eqref{betaexpansion} and \eqref{Haf-2}, we have
\begin{align*}
\langle 0|\beta_1\beta_2\cdots\beta_{2n}|0\rangle={\rm Haf} M.
\end{align*}
Let us summarize the results above into the proposition below.
\begin{proposition}\label{prop:Hf}\cite{luque}
Given $\beta_i\in V=\bigoplus_{j\in \mathbb{Z}+1/2}\mathbb{C}\phi_j$ with $\beta_i\beta_j=\beta_j\beta_i$,
\begin{align*}
\langle 0|\beta_1\beta_2\cdots\beta_{n}|0\rangle
=\left\{
\begin{array}{ll}
    0, & \hbox{${\rm if}\ n\ {\rm odd}$,} \\
{\rm Haf}\big(\langle0| \beta_i\beta_j|0\rangle\big)_{i,j=1}^n, & \hbox{${\rm if}\ n\ {\rm even}$.}
  \end{array}
\right.
\end{align*}
\end{proposition}

In order to give the bosonization of bosonic field $\phi(\lambda)$, we firstly introduce
\begin{align}
\sum\limits_{n\in2\mathbb{Z}+1}H_n \lambda^{-n-1}=-\frac{1}{2}:\phi(-\lambda)\phi(\lambda):, \ \
\sum_{j\in\mathbb{Z}+1/2}H_j\lambda^{-j-1/2}
=\frac{1}{2}V_-(\lambda)^{-1}\phi(\lambda)V_+(\lambda)^{-1},\label{hexpression}
\end{align}
where the normal ordering is defined by
\begin{align*}
:\phi_i\phi_j:=\left\{
\begin{array}{ll}
\phi_i\phi_j, & \hbox{${\rm if}\ i\geq j$,} \\
\phi_j\phi_i, & \hbox{${\rm if}\ j>i$,}
  \end{array}
\right.
\end{align*}
and
\begin{align*}
V_{\pm}(\lambda)=\exp\left(\sum_{\pm k>0,odd}\frac{2}{k}H_k\lambda^{-k}\right),
\end{align*}
then
\begin{align*}
H_iH_j-(-1)^{4ij}H_jH_i=\frac{j}{2}(-1)^{[j-\frac{1}{2}]}\delta_{i,-j}, \end{align*}
where $[j-\frac{1}{2}]$ denote the largest integer no more than $j-\frac{1}{2}$. It can be found that
\begin{align*}
H_k|0\rangle =\langle 0|H_{-k}=0, \ k>0.
\end{align*}
Further introduce
\begin{align*}
H(t_o):= \sum_{n\in2\mathbb{Z}_{\geq 0}+1}t_n H_n ,\ \ \chi(t_{s}):=\sum_{j\in\mathbb{Z}_{\geq 0}+\frac{1}{2}}t_j H_j,
\end{align*}
then
\begin{align*}
e^{H(t_o)}\phi(\lambda)e^{-H(t_o)}=e^{\xi(t_o,\lambda)}\phi(\lambda), \ \
e^{\chi(t_{s})}\phi(\lambda)e^{-\chi(t_{s})}=\phi(\lambda)^+V_-(\lambda)\varrho(\lambda)V_+(\lambda),
\end{align*}
where $\varrho(\lambda)=\sum_{k\in \mathbb{Z}_{\geq 0}+1/2}kt_k(-\lambda)^{k-\frac{1}{2}}$.

Based upon above preparation, now we can state the bosonization of $\mathcal{F}$ in the following proposition.
\begin{proposition}\label{prop:boson-fermion}
There is an isomorphism \cite{Anguelova2017,vandeleur2012,yang2021jmp}
\begin{eqnarray*}
\sigma:\ \mathcal{F}\rightarrow\mathcal{B}=\mathbb{C}[t_1,t_{\frac{1}{2}},t_3,t_{\frac{3}{2}}, \cdots ],\quad |u\rangle\mapsto\sigma(|u\rangle)=\langle 0|e^{H(t_o)}e^{\chi(t_{s})}|u\rangle.
\end{eqnarray*}
In particular
{\small\begin{align*}
\sigma \phi(\lambda) \sigma^{-1}=&\exp\left(\sum_{k\in 2\mathbb{Z}_{\geq 0}+1}t_k\lambda^k\right)
\exp\left(\sum_{k\in 2\mathbb{Z}_{\geq 0}+1}\frac{2}{k}\frac{\pa}{\pa t_k}\lambda^{-k}\right)
\sum_{j\in 2\mathbb{Z}_{\geq 0}+1}\left(\frac{j}{2}t_{\frac{j}{2}}(-\lambda)^{\frac{j-1}{2}}+2\frac{\pa}{\pa t_{\frac{j}{2}}}\lambda^{-\frac{j+1}{2}}\right).
\end{align*}}
\end{proposition}

\subsection{Bosonic construction of CKP hierarchy}
Firstly CKP hierarchy is corresponding to the infinite dimensional Lie algebra $c_{\infty}=\overline{c}_{\infty}\oplus\mathbb{C}c$ \cite{Anguelova2017,vandeleur2012,yang2021jmp}, where
\begin{align*}
\overline{c}_{\infty}=\{(a_{i,j})_{i,j\in \mathbb{Z}}| a_{i,j}=(-1)^{i+j+1}a_{-j+1,-i+1},a_{i,j}=0, \ {\rm for}\ |i-j|\gg0\}.
\end{align*}
The corresponding Lie brackets are given by
\begin{align*}
[A+\lambda c, B+\mu c]=[A,B]+w(A,B)c,\ \  \lambda,\mu\in \mathbb{C},
\end{align*}
where $w(A,B)$ is the 2--cocycle on $\overline{c}_{\infty}$ such that
\begin{align*}
w(E_{i,j})=-w(E_{j,i},E_{i,j})=1, \ \ i\leq0, j>0\ \text{and}\ w(E_{i,j},E_{k,l})=0,\ \ \ otherwise.
\end{align*}
Here $(E_{i,j})_{ab}=\delta_{ia}\delta_{jb}$. Note that the generator of\ $c_{\infty}$ is
\begin{align*}
Z_{i,j}=E_{i,j}-(-1)^{i+j}E_{-j+1,-i+1}.
\end{align*}
Therefore the representation $\widehat{\pi}$ of $c_{\infty}$ on the Fock space $\mathcal{F}$ can be defined by
\begin{align*}
\widehat{\pi}(Z_{ij})=(-1)^{j-1}:\phi_{i-\frac{1}{2}}\phi_{\frac{1}{2}-j}:,
\quad \widehat{\pi}(c)=\frac{1}{2}.
\end{align*}

Note that by \eqref{intro:sp},
\begin{align*}
Sp_{\infty}=\{e^{\widehat{\pi}(a_1)}e^{\widehat{\pi}(a_2)}\cdots e^{\widehat{\pi}(a_k)}|a_i\in c_\infty\}.
\end{align*}
It can be proved that for $g\in Sp_{\infty}$,
\begin{align*}
S(g\otimes g)=(g\otimes g)S.
\end{align*}
Then if set $\tau=g|0\rangle$, we will have
\begin{align}\label{S-tau}
S(\tau\otimes \tau)=0,
\end{align}
which is called the bosonic CKP hierarchy.
\begin{remark}
Here for $\tau$ in \eqref{S-tau},  $\tau\in \mathcal{\overline{F}}$ with $\overline{\mathcal{F}}$ being the completion of $\mathcal{F}$\cite{Anguelova2017}. And
\eqref{S-tau} is sub--hierarchy of super--BKP hierarchy in \cite{kac1989}.
\end{remark}
If set
$\tau(t_o,t_s)=\sigma(\tau)=\langle | e^{H(t_o)}e^{\chi(t_s)}g|0\rangle$,
then
\begin{align}\label{super-CKP}
{\rm Res}_\lambda\sum_{i,j\in \mathbb{Z}_{\geq 0}+\frac{1}{2}}
e^{\xi(t_o-t_o',\lambda)}&e^{\xi(\widetilde{\pa}_{t_o}-\widetilde{\pa}_{t'_o},\lambda^{-1})}
\left(\frac{\pa}{\pa t_i}\lambda^{-i-\frac{1}{2}}-2it_i(-\lambda)^{i-\frac{1}{2}}\right)\nonumber\\
&\times\left(\frac{\pa}{\pa t'_j}(-\lambda)^{-j-\frac{1}{2}}-2jt'_j\lambda^{j-\frac{1}{2}}\right)
\tau(t_o,t_s)\tau(t'_o,t'_s)=0.
\end{align}
Further if expand $\tau(t_o,t_s)$ in terms of super variables
\begin{align*}
\tau(t_o,t_s)=\tau_0(t_o)+\sum_{\nu\in\mathbb{Z}^{>1}_{odd}}\tau_{(\nu,1)}(t_o)t_{\frac{\nu}{2}}t_{\frac{1}{2}}+{\rm terms\ of\ at\ least\ two}\ t_s.
\end{align*}
then
\begin{align*}
\tau_{0}(t_o)=\langle 0|e^{H(t_o)}g|0\rangle,\ \
\tau_{(\nu,1)}(t_o)=4\langle 0|H_{\frac{\nu}{2}}H_{\frac{1}{2}}e^{H(t_o)}g|0\rangle,\ \cdots.
\end{align*}
It can be found that the coefficients of $t_{\frac{1}{2}}t'_{\frac{1}{2}}$ in \eqref{super-CKP} will give rise to the bilinear equation of CKP hierarchy
\begin{align*}
{\rm Res}_\lambda\psi(t_o,\lambda)\psi(t'_o,-\lambda)=0,
\end{align*}
where
\begin{align*}
\psi(t_o,\lambda)=\frac{\tau_0(t_o+2[\lambda^{-1}])+ \sum\limits_{\nu\in\mathbb{Z}^{>1}_{odd}}\tau_{(\nu,1)}(t_o+2[\lambda^{-1}])\lambda^{-\frac{\nu+1}{2}}}
{\tau_0(t_o)}e^{\xi(t_o,\lambda)}.
\end{align*}

Note that by Proposition \ref{prop:boson-fermion}
\begin{align}
\psi(t_o,\lambda)=\frac{2\langle 1|e^{H(t_o)}\phi(\lambda)g|0\rangle}{\langle 0|e^{H(t_o)}g|0\rangle},\label{ckp-wave-bosonic-vev}
\end{align}
where $\langle 1|=\langle 0|H_{\frac{1}{2}}$. If assume
\begin{align*}
\psi(t_o,\lambda)=(1+w_1\lambda^{-1}+w_2\lambda^{-2}+\cdots)e^{{\xi}(t_o,\lambda)},
\end{align*}
then
\begin{align*}
w_1=\frac{2\pa_x\tau_0(t_o)}{\tau_0(t_o)}=2\pa_x\log\tau_0(t_o),\ \ w_2=\frac{2\pa^2_x\tau_0(t_o)}{\tau_0(t_o)}+4\frac{\tau_{(3,1)}(t_o)}{\tau_0(t_o)}.
\end{align*}
Comparing with
\begin{align*}
\psi(t_o,\lambda)&=\left(1+\frac{1}{\lambda}\pa_x{\rm log}\frac{\tau_{{\rm CW}}(t_o-2[\lambda^{-1}])}{\tau_{\rm CW}(t_o)}\right)^{\frac{1}{2}}
\frac{\tau_{\rm CW}(t_o-2[\lambda^{-1}])}{\tau_{\rm CW}(t_o)}e^{\xi(t_o,\lambda)},
\end{align*}
we have
\begin{align*}
&\tau_0=\tau_{{\rm CW}}^{-1},\\
&\tau_{(3,1)}=\frac{3}{4}(\log\tau_{{\rm CW}})_{xx}\tau^{-1}_{{\rm CW}},\\
&\tau_{(5,1)}=-\frac{5}{4}(\log\tau_{{\rm CW}})_{xxx}\tau^{-1}_{{\rm CW}},\\
&\tau_{(7,1)}=\frac{7}{8}(\log\tau_{{\rm CW}})^2_{xx}\tau^{-1}_{{\rm CW}}
+\frac{7}{12}(\log\tau_{{\rm CW}})_{xt_3}\tau^{-1}_{{\rm CW}}
+\frac{7}{6}(\log\tau_{{\rm CW}})_{xxxx}\tau^{-1}_{{\rm CW}}.
\end{align*}

\subsection{Proof of Theorem \ref{Th:CW-van}}
Given $\beta\in V=\bigoplus_{j\in \mathbb{Z}+1/2}\mathbb{C}\phi_j$, there exists $\rho(\lambda)\in \mathbb{C}((\lambda^{-1}))$ such that \begin{align*}
\beta={\rm Res}_\lambda\rho(\lambda)\phi(\lambda).
\end{align*}
Thus if denote
\begin{align*}
q(t_o)=\frac{2\langle 1|e^{H(t_o)}\beta g|0\rangle}{\langle 0|e^{H(t_o)}g|0\rangle},
\end{align*}
then
$$q(t_o)={\rm Res}_\lambda\rho(\lambda)\psi(t_o,\lambda),$$
where $\psi(t_o,\lambda)=\frac{2\langle 1|e^{H(t_o)}\phi(\lambda) g|0\rangle}{\langle 0|e^{H(t_o)}g|0\rangle}$ is the CKP wave function. So $q(t_o)$ is the CKP eigenfunction corresponding to CKP wave function $\psi(t_o,\lambda)$, that is
\begin{align*}
q_{t_n}=(W\pa^nW^{-1})_{\geq0}(q),
\end{align*}
where $W=1+\sum_{j= 1}^{\infty}w_j\pa^{-j}$ such that $\psi(t_o,\lambda)=W(e^{\xi(t_o,\lambda)})$.

In order to prove Theorem \ref{Th:CW-van}, we need firstly to prove Proposition \ref{Prop:SEP-vac}.
\begin{proof}
Firstly for $\beta\in V$, we have by Lemma \ref{lemma:S-beta}
\begin{align}\label{S-tau-beta}
S(\tau\otimes\beta\tau)=-\beta\tau\otimes\tau.
\end{align}
If apply $\langle 1|e^{H(t_o)}\otimes\langle 0|e^{H(t'_o)}$ to both sides of \eqref{S-tau-beta}, we can get
\begin{align*}
{\rm Res}_\lambda\psi(t_o,\lambda)\frac{\langle 0|e^{H(t'_o)}\phi(-\lambda)\beta g|0\rangle}{\langle 0|e^{H(t'_o)}\beta g|0\rangle}=-q(t_o),
\end{align*}
where $q(t_o)=\frac{2\langle 1|e^{H(t_o)}\beta g|0\rangle}{\langle 0|e^{H(t_o)}g|0\rangle}$.
Further by \eqref{CKP-SEP-q}
\begin{align*}
{\rm Res}_\lambda\psi(t_o,\lambda)
\left(\Omega(q(t'_o),\psi(t'_o,-\lambda)-
\frac{\langle 0|e^{H(t'_o)}\phi(\lambda)\beta g|0\rangle}{\langle 0|e^{H(t'_o)}\beta g|0\rangle}\right)=0.
\end{align*}
Next by \eqref{KP-SEP-q-O} or \eqref{KP-SEP-r-O}, we can find
\begin{align*}
\Omega(q(t'_o),\psi(t'_o,-\lambda)-
\frac{\langle 0|e^{H(t'_o)}\phi(\lambda)\beta g|0\rangle}{\langle 0|e^{H(t'_o)}\beta g|0\rangle}
=\mathcal{O}(\lambda^{-2})e^{-\xi(t'_o,\lambda)}.
\end{align*}
Thus by the following fact \cite{Date1983WorldSci}:\\

\fbox{\textit{
If
${\rm Res}_\lambda\psi(t,\lambda)\widetilde{\psi}(t',\lambda)=0$
for $
\widetilde{\psi}(t,\lambda)=\mathcal{O}(\lambda^{-1})e^{-\widetilde{\xi}(t,\lambda)},$
then $
\widetilde{\psi}(t,\lambda)=0.$}
}

\ \\
\noindent We can finally know
\begin{align*}
\Omega(q(t_o),\psi(t_o,-\lambda))=
\frac{\langle 0|e^{H(t_o)}\phi(\lambda)\beta g|0\rangle}{\langle 0|e^{H(t_o)}\beta g|0\rangle}.
\end{align*}

As for $\beta_1,\beta_2\in V$, there exists $\rho_2(\lambda)$ such that
\begin{align*}
\beta_2={\rm Res}_\lambda\rho_2(\lambda)\phi(\lambda),
\end{align*}
which implies that
\begin{align*}
q_2(t_o)={\rm Res}_\lambda\rho_2(\lambda)\psi(t_o,\lambda).
\end{align*}
Therefore
\begin{align*}
\Omega(q_1(t_o),q_2(t_o))={\rm Res}_\lambda\rho_2(\lambda)\Omega(q_1(t_o),\psi(t_o,\lambda))=
\frac{\langle 0| e^{H(t_o)}\beta_1\beta_2g |0\rangle}{\langle 0| e^{H(t_o)}g |0\rangle}.
\end{align*}
\end{proof}
For the more general case of Proposition \ref{Prop:SEP-vac}, we have the following result.
\begin{proposition}\label{prop:Hf-fr}
Given $\beta\in V,\  g\in Sp_{\infty}$ with $\beta_i\beta_j=\beta_j\beta_i$
\begin{align*}
\frac{\langle 0|e^{H(t_o)}\beta_1\beta_2\cdots\beta_{n}g|0\rangle}
{\langle 0|e^{H(t_o)}g|0\rangle}
=\left\{
\begin{array}{ll}
    0, & \hbox{${\rm if}\ n\ {\rm odd}$,} \\
{\rm Haf}\left(\frac{\langle0|e^{H(t_o)} \beta_i\beta_jg|0\rangle}{\langle 0|e^{H(t_o)}g|0\rangle}\right)_{i,j=1}^n, & \hbox{${\rm if}\ n\ {\rm even}$.}
  \end{array}
\right.
\end{align*}
\end{proposition}
\begin{proof}
Similar to Proposition \ref{prop:Hf}, one only need to know
\begin{align*}
\Big(\langle 0|e^{H(t_o)}\otimes\langle 0|e^{H(t_o)}\Big)S=0,
\end{align*}
and $S(g|0\rangle\otimes g|0\rangle)=0$, which can be derived by $S(g\otimes g)=(g\otimes g)S$.
\end{proof}
To prove Theorem \ref{Th:CW-van}, the following lemma is also needed.
\begin{lemma}\label{lemma:SEp-vac-2}
Given $\beta\in V, g\in Sp_{\infty}$,
\begin{align*}
\frac{\langle 0|e^{H(t_o)}e^{\frac{\beta^2}{2}} g|0\rangle}{\langle 0|e^{H(t_o)} g|0\rangle}=\left(1-\frac{\langle 0|e^{H(t_o)}\beta^2 g|0\rangle}{\langle 0|e^{H(t_o)} g|0\rangle}\right)^{-\frac{1}{2}}.
\end{align*}
\end{lemma}
\begin{proof}
Firstly by \eqref{hafdefinition} and Proposition \ref{prop:Hf-fr}, we can know
\begin{align*}
\frac{\langle 0|e^{H(t_o)}\beta^{2m} g|0\rangle}{\langle 0|e^{H(t_o)} g|0\rangle}=(2m-1)!!\left(\frac{\langle 0|e^{H(t_o)}\beta^2 g|0\rangle}{\langle 0|e^{H(t_o)} g|0\rangle}\right)^m.
\end{align*}
Then by $(1-\lambda)^{-\frac{1}{2}}=\sum^{\infty}_{m=0}\frac{(2m-1)!!}{m!2^m}\lambda^m$,
one can find
\begin{align*}
\frac{\langle 0|e^{H(t_o)}e^{\frac{\beta^2}{2}} g|0\rangle}{\langle 0|e^{H(t_o)} g|0\rangle}
=\sum_{m=0}^{\infty}\frac{1}{m!2^m}
\frac{\langle 0|e^{H(t_o)}\beta^{2m} g|0\rangle}{\langle 0|e^{H(t_o)} g|0\rangle}
=\left(1-\frac{\langle 0|e^{H(t_o)}\beta^2 g|0\rangle}{\langle 0|e^{H(t_o)} g|0\rangle}\right)^{-\frac{1}{2}}.
\end{align*}
\end{proof}

After the preparation above, now let us prove Theorem \ref{Th:CW-van}.
\begin{proof}
Firstly when $n=1$, by Lemma \ref{lemma:SEp-vac-2} we can know
\begin{align*}
\langle 0|e^{H(t_o)}e^{\frac{\beta^2}{2}} |0\rangle^{-1}
=\left(1-\langle 0|e^{H(t_o)}\beta^2|0\rangle\right)^{\frac{1}{2}},
\end{align*}
while by Proposition \ref{Prop:SEP-vac}, we can know that
\begin{align*}
\langle 0|e^{H(t_o)}\beta^2 |0\rangle
=\Omega(q_0(t_o),q_0(t_o)),
\end{align*}
where $q_0(t_o)=\langle 1|e^{H(t_o)}\beta |0\rangle$ is the CKP eigenfunction corresponding CKP tau function $\tau_{{\rm CW}}=1$. Thus
$$\langle 0|e^{H(t_o)}e^{\frac{\beta^2}{2}} |0\rangle^{-1}=
\Big(1-\Omega(q_0(t_o),q_0(t_o))\Big)^{\frac{1}{2}},$$
is CKP tau function by Theorem \ref{Th:darboux-bili}. Assume we have prove $\langle 0|e^{H(t_o)}g |0\rangle^{-1}$ is CKP tau function with $$g=e^{\frac{{\beta}_{n-1}^2}{2}}\cdots
e^{\frac{{\beta}_{1}^2}{2}},$$
then by Proposition \ref{Prop:SEP-vac} and Lemma \ref{lemma:SEp-vac-2},
\begin{align}\label{vac-eH-g}
\langle 0|e^{H(t_o)}e^{\frac{\beta^2}{2}}g |0\rangle^{-1}
=\Big(1-\Omega(q_0(t_o),q_0(t_o))\Big)^{\frac{1}{2}}
\langle 0|e^{H(t_o)}g |0\rangle^{-1},
\end{align}
is still a CKP tau function due to Theorem \ref{Th:darboux-bili}, where
$q(t_o)=\frac{2\langle 1|e^{H(t_o)}\beta g|0\rangle}{\langle 0|e^{H(t_o)}g|0\rangle}$ is the CKP eigenfunction corresponding to CKP tau function $\langle 0|e^{H(t_o)} g|0\rangle^{-1}$.
\end{proof}
As for Corollary \ref{coro:CKP-tau}, the corresponding proof is given as follows.
\begin{proof}
Firstly by Theorem \ref{Th:darboux-bili}, \eqref{vac-eH-g} and Remark \ref{remark:taucw-Darboux}, we can know that
$$\tau_{\rm CW}^{\{1\}}(t_o)=\langle0|e^{H(t_o)}e^{\frac{\beta_{1}}{2}}g|0\rangle^{-\frac{1}{2}}
=\Big(1-\Omega(q_{1}(t_o),q_{1}(t_o))^{\frac{1}{2}}\Big)\tau_{\rm CW}(t_o),$$
is the transformed CKP tau function under CKP Darboux transformation $T[q_{1}]$.
Next by Remark \ref{remark:taucw-Darboux}, the transformed wave function $\psi^{\{1\}}(t_o,\lambda)$ will be
$$\psi^{\{1\}}(t_o,\lambda)=\frac{2\langle1|e^{H(t_o)}\phi(\lambda)e^{\frac{\beta_{1}^2}{2}}g|0\rangle}
{\langle0|e^{H(t_o)}e^{\frac{\beta_{1}^2}{2}}g|0\rangle}=T[q_1](\psi(t_o,\lambda)).$$
Note that there exists $\rho(\lambda)$ such that $\beta_2={\rm Res}_\lambda\rho(\lambda)\phi(\lambda)$, thus
$q_2(t_o)={\rm Res}_\lambda\rho(\lambda)\psi(t_o,\lambda)$. Finally
\begin{align*}
\frac{2\langle0|e^{H(t_o)}\beta_2e^{\frac{\beta_{1}^2}{2}}g|0\rangle}
{\langle0|e^{H(t_o)}e^{\frac{\beta_{1}^2}{2}}g|0\rangle}
={\rm Res}_\lambda\rho(\lambda)\frac{2\langle0|e^{H(t_o)}\phi(\lambda)e^{\frac{\beta_{1}^2}{2}}g|0\rangle}
{\langle0|e^{H(t_o)}e^{\frac{\beta_{1}^2}{2}}g|0\rangle}=T[q_1](q_2).
\end{align*}
\end{proof}

\subsection{Proof of Theorem \ref{Them:k-darboux}}
\begin{proof}
According to Corollary \ref{coro:CKP-tau}, we can know that
$$\langle 1|e^{H(t_o)}e^{\frac{\beta_n^2}{2}}
e^{\frac{\beta_{n-1}^2}{2}}\cdots
e^{\frac{\beta_1^2}{2}}g|0\rangle^{-1},$$
is the transformed CKP tau function under the CKP Darboux transformation
\begin{align*}
T[q_n^{\{n-1\}}]\cdots T[q_2^{\{1\}}]T[q_1],
\end{align*}
starting from $\langle 0|e^{H(t_o)}g|0\rangle^{-1}$, where
\begin{align*}
q_i=\frac{2\langle0|e^{H(t_o)}\beta_ig|0\rangle}{\langle0|e^{H(t_o)}g|0\rangle}.
\end{align*}
On the other hand by Lemma \ref{lemma:SEp-vac-2}.
\begin{align*}
\langle0|e^{H(t_o)}e^{\frac{\beta_n^2}{2}}e^{\frac{\beta_{n-1}^2}{2}}\cdots e^{\frac{\beta_1^2}{2}}g|0\rangle
=\prod_{k=1}^n\left(1-\frac{\langle0|e^{H(t_o)}\beta_k^2e^{\frac{\beta_{k-1}^2}{2}}\cdots e^{\frac{\beta_{1}^2}{2}}g|0\rangle}
{\langle0|e^{H(t_o)}e^{\frac{\beta_{k-1}^2}{2}}\cdots e^{\frac{\beta_{1}^2}{2}}g|0\rangle}\right)^{-\frac{1}{2}}
\langle0|e^{H(t_o)}g|0\rangle,
\end{align*}
while according to Proposition \ref{Prop:SEP-vac} and Corollary \ref{coro:CKP-tau}
\begin{align*}
&\frac{\langle0|e^{H(t_o)}\beta_k^2e^{\frac{\beta_{k-1}^2}{2}}\cdots e^{\frac{\beta_1^2}{2}}g|0\rangle}
{\langle0|e^{H(t_o)}e^{\frac{\beta_{k-1}^2}{2}}\cdots e^{\frac{\beta_1^2}{2}}g|0\rangle}-1
=\Omega(q_k^{\{k-1\}},q_k^{\{k-1\}})-1\\
=&\left(\Omega(q_k^{\{k-2\}},q_k^{\{k-2\}})-1\right)
-\left(\Omega(q_{k-1}^{\{k-2\}},q_{k-1}^{\{k-2\}})-1\right)^{-1}
\cdot\Omega(q_{k-1}^{\{k-2\}},q_k^{\{k-2\}})\cdot\Omega(q_k^{\{k-2\}},q_{k-1}^{\{k-2\}}).
\end{align*}
Therefore similar to proof of Corollary \ref{coro:KP-Darboux-n},
\begin{align*}
&\left(\Omega(q_n^{\{n-1\}},q_n^{\{n-1\}})-1\right)\cdots
\left(\Omega(q_2^{\{1\}},q_2^{\{1\}})-1\right)
\Big(\Omega(q_1,q_1)-1\Big)\\
&=\left|\begin{array}{cccc}
\Omega_{11}-1  &\Omega_{12}&\cdots &\Omega_{1n}\\
\Omega_{21} &\Omega_{22}-1&\cdots &\Omega_{2n}\\
\cdots & \cdots &\cdots &\cdots\\
\Omega_{n1}&\Omega_{n2}&\cdots &\Omega_{nn}-1
\end{array}
\right|,
\end{align*}
where $$\Omega_{ij}=\Omega(q_i,q_j)=\frac{\langle0|e^{H(t_o)}\beta_i\beta_jg|0\rangle}
{\langle0|e^{H(t_o)}g|0\rangle}.$$
Thus we can prove Theorem \ref{Them:k-darboux}.
\end{proof}

\noindent{\bf Example:}\\
Let us compute $\langle0|e^{H(t_o)}e^{\frac{a\phi(p)\phi(q)}{2}}|0\rangle$. For the usual way, we firstly need to expand $$e^{\frac{a\phi(p)\phi(q)}{2}}=\sum_{k=0}^{\infty}\frac{1}{k!}\left(\frac{a\phi(p)\phi(q)}{2}\right)^k,$$
 and then compute $\langle0|(\phi(p)\phi(q))^k|0\rangle$ by Proposition \ref{prop:Hf}, that is,
\begin{align*}
\langle0|\big(\phi(p)\phi(q)\big)^2|0\rangle=&\langle\phi^2(p)\rangle\langle\phi^2(q)\rangle+2\langle\phi(p)\phi(q)\rangle^2
=\frac{1}{4pq}+\frac{2}{(p+q)^2},\\
\langle0|\big(\phi(p)\phi(q)\big)^3|0\rangle=&9\langle\phi^2(p)\rangle\langle\phi(p)\phi(q)\rangle\langle\phi^2(q)\rangle+6\langle\phi(p)\phi(q)\rangle^3
=\frac{9}{4pq(p+q)}+\frac{6}{(p+q)^3},\\
\langle0|\big(\phi(p)\phi(q)\big)^4|0\rangle=&9\langle\phi^2(p)\rangle^2\langle\phi^2(q)\rangle^2+72\langle\phi^2(p)\rangle\langle\phi(p)\phi(q)\rangle^2\langle\phi^2(q)\rangle
+24\langle\phi(p)\phi(q)\rangle^4\\
=&\frac{9}{16p^2q^2}+\frac{18}{pq(p+q)^2}+\frac{24}{(p+q)^4},\\
\cdots\ \ \cdots& \ \ \cdots
\end{align*}
Note that it is quite difficult to find a general formula for $\langle0|(\phi(p)\phi(q))^k|0\rangle$. But we can compute $\langle0|e^{H(t_o)}e^{\frac{a\phi(p)\phi(q)}{2}}|0\rangle$ easily according to Theorem \ref{Them:k-darboux}. In fact due to
\begin{align*}
\phi(p)\phi(q)-\phi(q)\phi(p)=\delta(p,-q).
\end{align*}
Thus when $p\neq -q$, we can believe $\phi(p)$ commutes with $\phi(q)$.
If set $$\beta_1=\frac{\sqrt{-a}(\phi(p)-\phi(q))}{2}, \beta_2=\frac{\sqrt{a}(\phi(p)+\phi(q))}{2},$$
then $e^{\frac{a\phi(p)\phi(q)}{2}}
=e^{\frac{\beta_2^2}{2}}e^{\frac{\beta_1^2}{2}}$. Note that
\begin{align*}
&\Omega_{11}=\langle0|e^{H(t_o)}\beta_1^2|0\rangle
=-\frac{a}{4}\left(\frac{e^{2\xi(t_o,p)}}{2p}+
\frac{e^{2\xi(t_o,q)}}{2q}
-\frac{2e^{\xi(t_o,p)+\xi(t_o,q)}}{p+q}
\right),\\
&\Omega_{22}=\langle0|e^{H(t_o)}\beta_2^2|0\rangle
=\frac{a}{4}\left(\frac{e^{2\xi(t_o,p)}}{2p}
+\frac{e^{2\xi(t_o,q)}}{2q}
+\frac{2e^{\xi(t_o,p)+\xi(t_o,q)}}{p+q}
\right),\\
&\Omega_{12}=\langle0|e^{H(t_o)}\beta_1\beta_2|0\rangle
=\frac{a\sqrt{-1}}{4}\left(\frac{e^{2\xi(t_o,p)}}{2p}
-\frac{e^{2\xi(t_o,q)}}{2q}
\right)=\Omega_{21}.
\end{align*}
Finally by Theorem \ref{Them:k-darboux},
\begin{align*}
\langle0|e^{H(t_o)}e^{\frac{a\phi(p)\phi(q)}{2}}|0\rangle
=\left(1-\frac{ae^{\xi(t_o,p)+\xi(t_o,q)}}{p+q}
-\frac{a^2(p-q)^2}{16pq(p+q)^2}e^{2\xi(t_o,p)+2\xi(t_o,q)}
\right)^{-\frac{1}{2}}.
\end{align*}

\section{KdV and CKP hierarchies}
In this section, we firstly show that KdV hierarchy is the 2--reduction of CKP hierarchy, and then construct solutions of KdV hierarchy by vacuum expectation value of bosonic fields.

\subsection{KdV hierarchy is 2--reduction of CKP hierarchy}
KdV hierarchy is defined by the following Lax equation
\begin{align*}
\mathfrak{L}_{t_n}=[\mathfrak{B}_n,\mathfrak{L}],\ \mathfrak{B}_n=(\mathfrak{L}^{\frac{n}{2}})_{\geq 0},\  n=1,3,5,\cdots,
\end{align*}
with Lax operator
\begin{align}\label{KdVLax-operator}
\mathfrak{L}=\pa^2+u.
\end{align}
Because of \eqref{KdVLax-operator}, we can find
\begin{align*}
(\mathfrak{L}^2)_{<0}=0,\ \ \mathfrak{L}^*=\mathfrak{L}.
\end{align*}
Then we will prove that KdV hierarchy is just 2--reduction of CKP hierarchy.

Firstly if assume
$$\mathcal{A}=L_{t_n}-[(L^n)_{\geq 0},L]=\sum_{i=1}^{\infty}a_i\pa^{-i},$$
then
\begin{align*}
(L_{t_n})^2=[(L^n)_{\geq 0},L^2]+\mathcal{A}L+L\mathcal{A}.
\end{align*}
Since $L^2=\mathfrak{L}$ satisfies the KdV hierarchy, thus
\begin{align*}
\mathcal{A}L+L\mathcal{A}=2a_1+\mathcal{O}(\pa^{-1})=0,
\end{align*}
which implies that $a_1=0$.
If assume that we have proved $\mathcal{A}=a_k\pa^{-k}+\sum\limits_{i>k}^{\infty}a_i\pa^{-i}$, then by
\begin{align*}
\mathcal{A}L+L\mathcal{A}=2a_{k}\pa^{-k+1}+\mathcal{O}(\pa^{-k})=0,
\end{align*}
we can obtain $a_{k}=0$. Thus $\mathcal{A}=0$, that is $L_{t_n}=[(L^n)_{\geq 0},L]$.

Next let us assume
\begin{align*}
\mathcal{B}=L^*+L=\sum_{j=2}^\infty b_j\pa^{-j}.
\end{align*}
Thus
\begin{align*}
(L^*)^2=L^2-L\mathcal{B}-\mathcal{B}L+\mathcal{B}^2.
\end{align*}
By $\mathfrak{L}^*=\mathfrak{L}$, that is $(L^*)^2=L^2$, we can obtain
\begin{align}\label{LB-BL}
L\mathcal{B}+\mathcal{B}L=\mathcal{B}^2,
\end{align}
which tells that
\begin{align*}
2b_2\pa^{-1}+\mathcal{O}(\pa^{-2})=b_2^2\pa^{-2}+\mathcal{O}(\pa^{-3})\Rightarrow b_2=0.
\end{align*}
If we have proved $b_l=0(2\leq l \leq k-1)$, the \eqref{LB-BL} will be
\begin{align*}
2b_k\pa^{-k+1}+\mathcal{O}(\pa^{-k})=b_k^2\pa^{-2k}+\mathcal{O}(\pa^{-2k-1}),\ \ k\geq2.
\end{align*}
Note that $-k+1\geq 2k$, thus $b_k=0$, which implies that $\mathcal{B}=0$, that is  $L^*=-L$.

By now, we have proved the statement that KdV hierarchy is the 2--reduction of CKP hierarchy \cite{chang2013,Xie-arXiv} from the Lax representation, which is just the first part of Theorem \ref{Th:SEP-vac-2}. Note that the condition $\mathfrak{L}^*=\mathfrak{L}$ for KdV Lax operator $\mathfrak{L}$ is the key in the proof.

As a sub--hierarchy of KP hierarchy. KdV tau function
\begin{align*}
\tau_{{\rm KdV}}(t_o)=\tau_{{\rm KP}}(\dot{t}),
\end{align*}
where $\tau_{{\rm KP}}$ does not depend on $t_e=(t_2,t_4,\cdots)$. Note that the adjoint KdV wave function
\begin{align*}
\psi^*_{{\rm KdV}}(t_o,\lambda)=\frac{\tau_{{\rm KdV}}(\dot{t}+[\lambda^{-1}])}{\tau_{{\rm KdV}}(\dot{t})}e^{-\xi(t_o,\lambda)}=\frac{\tau_{{\rm KP}}(\dot{t}+[-\lambda^{-1}])}{\tau_{{\rm KP}}(\dot{t})}e^{-\xi(t_o,\lambda)}
=\psi_{{\rm KdV}}(t_o,-\lambda).
\end{align*}
Thus the bilinear equation of KdV hierarchy
\begin{align*}
&{\rm Res}_\lambda
\psi_{{\rm KdV}}(t_o,\lambda)\psi^*_{{\rm KdV}}(t'_o,\lambda)=0,\\
&{\rm Res}_\lambda \lambda^2
\psi_{{\rm KdV}}(t_o,\lambda)\psi^*_{{\rm KdV}}(t'_o,\lambda)=0,
\end{align*}
one can written into \cite{chang2013,vandeleur-CKP-2023}
\begin{align*}
&{\rm Res}_\lambda
\psi_{{\rm KdV}}(t_o,\lambda)\psi_{{\rm KdV}}(t'_o,-\lambda)=0,\\
&{\rm Res}_\lambda \lambda^2
\psi_{{\rm KdV}}(t_o,\lambda)\psi_{{\rm KdV}}(t'_o,-\lambda)=0,
\end{align*}
which shows again that the KdV hierarchy is the 2--reduction of the CKP hierarchy.

\subsection{KdV solutions by free boson}
Let us construct the solutions of KdV hierarchy from the aspect of CKP hierarchy. For this, introduce the following operator $S_2$ on $\mathcal{F}\otimes \mathcal{F}$,
\begin{align*}
S_2={\rm Res}_\lambda \lambda^2
\phi(\lambda)\otimes\phi(-\lambda)
=\sum_{k\in \mathbb{Z}+\frac{1}{2}}(-1)^{k+\frac{1}{2}}\phi_k\otimes\phi_{-k-2}.
\end{align*}
Then for $g\in Sp_{\infty}$, if
\begin{align}\label{s2-g}
[S_2,g\otimes g]=0,
\end{align}
then $\tau=g|0\rangle$ satisfies
\begin{align*}
S(\tau\otimes \tau)=S_2(\tau\otimes\tau)=0.
\end{align*}
Further if apply $\langle 1| e^{H(t_o)}\otimes \langle 1|e^{H(t_o)}$ on above relations and use $\psi(t_o,\lambda)=\frac{2\langle 1|e^{H(t_o)}\phi(\lambda) g|0\rangle}{\langle 0|e^{H(t_o)}g|0\rangle}$, then we can obtain
\begin{align*}
&{\rm Res}_\lambda
\psi(t_o,\lambda)\psi(t'_o,-\lambda)=0,\\
&{\rm Res}_\lambda \lambda^2
\psi(t_o,\lambda)\psi(t'_o,-\lambda)=0,
\end{align*}
which is just the 2--reduction of CKP hierarchy \cite{chang2013,vandeleur-CKP-2023}. Therefore for $g\in Sp'_\infty$ satisfying \eqref{s2-g},
\begin{align*}
\tau_{{\rm CW}}(t_o)=\langle 0|e^{H(t_o)}g|0\rangle^{-1},
\end{align*}
will be the solution of the KdV hierarchy as the CKP hierarchy. If assume $g=\sum_{i\in \mathbb{Z}+1/2} a_{ij}\phi_i\phi_j$ satisfying \begin{align*}
[S_2,1\otimes g+g\otimes 1]=0,
\end{align*}
which implies that
\begin{align}\label{aij-aji}
a_{ij}+a_{ji}=a_{i+2,j-2}+a_{j-2,i+2},
\end{align}
then $g=\sum_{i\in \mathbb{Z}+1/2} a_{ij}\phi_i\phi_j$ will satisfy \eqref{s2-g}, thus we complete the proof of Theorem \ref{Th:SEP-vac-2}.\\

\noindent{\bf Example:}\\
If we choose $a_{ij}=\frac{1}{2}p^{i+j-1}$, then $a_{ij}$ satisfies \eqref{aij-aji} and $\sum_{i,j\in \mathbb{Z}+1/2} a_{ij}\phi_i\phi_j=\frac{1}{2}\phi^2(p)$. Therefore
\begin{align*}
\tau_{{\rm CW}}(t_o)=\langle 0|e^{H(t_o)}e^{\frac{1}{2}\phi^2(p)}|0\rangle^{-1}
=\left(1-\frac{1}{2p}e^{2\xi(t_o,p)}\right)^{\frac{1}{2}},
\end{align*}
is the solution of KdV hierarchy as CKP hierarchy. The more general case can be considered by Theorem \ref{Them:k-darboux}
\begin{align*}
\tau_{{\rm CW}}(t_o)&=
\langle 0| e^{H(t_o)}e^{\frac{\phi^2(p_n)}{2}}
e^{\frac{\phi^2(p_{n-1})}{2}}\cdots e^{\frac{\phi^2(p_1)}{2}} |0\rangle^{-1}\\
&=(-1)^{-\frac{n}{2}}\begin{vmatrix}
A_{11}-1 &A_{12}&\cdots&A_{1n}\\
A_{21}&A_{22}-1&\cdots&A_{2n}\\
\cdots&\cdots&\cdots&\cdots\\
A_{n1}&A_{n2}&\cdots&A_{nn}-1
\end{vmatrix}^{-\frac{1}{2}},
\end{align*}
where
\begin{align*}
A_{ij}=\langle 0|e^{H(t_o)}\phi(p_i)\phi(p_j)|0\rangle
=\frac{2e^{\xi(t_o,p_i)+\xi(t_o,p_j)}}{p_i+p_j}.
\end{align*}

\section{Conclusions and Discussion}
The major topic of this paper is the CKP tau function. Firstly, the transformed CKP tau function under the CKP Darboux transformation is showed to still satisfy the CKP bilinear equation via wave functions. Based upon this, the inverse of vacuum expectation value of the exponential of certain bosonic fields, is showed to be the CKP tau function given by Chang and Wu in \cite{chang2013}, which is the major theorem of this paper. Next by using the result of successive applications of CKP Darboux transformation, we give a formula to compute the vacuum expectation value of some bosonic fields. Finally, KdV hierarchy is proved to be the 2--reduction of CKP hierarchy and the corresponding solutions are constructed by vacuum expectation value of bosonic fields.

We believe the results here are very basic in the study of CKP hierarchy, which provide the theoretical support that we can investigate CKP hierarchy by free bosons.  For further discussion, we can use free bosons to investigate bilinear equations in successive applications of CKP Darboux transformations, and study the symmetries of CKP hierarchy. What's more, other bosonizations \cite{Anguelova2017,vandeleur2023} of free bosons are also very interesting, for example the mutli--component generalization of CKP \cite{vandeleur2023,Zabrodin2023}. Also some symmetric functions and other integrable models related with CKP hierarchy may be considered by using free bosons.\\

\section{Declarations}
\noindent{\bf Acknowledgements}: \\
We thank Profs. Jingsong He (SZU), Zhiwei Wu (SYSU), Shihao Li (SCU) and Dr. Yuancheng Xie (PKU) for their comments and assistance.
This work is supported by National Natural Science Foundation of China
(Grant Nos. 12171472 and 12261072) and ``Qinglan Project" of Jiangsu Universities.\\

\noindent{\bf Data availability}: \\
Date sharing is not applicable to this article as no new data were created or analyzed in this study.\\

\noindent{\bf Conflicts of interests/Competing interests}: \\
The authors have no conflicts of interests/competing interests to declare that are relevant to the content of this article.

\end{document}